\documentclass[a4paper,reqno]{article}

    \usepackage[english]{babel} 

    \usepackage{graphicx}
    \usepackage[]{color}  
    \usepackage[dvipsnames]{xcolor} 
    \usepackage{csquotes} 
    \usepackage{appendix} 
    \usepackage[a4paper]{geometry} 
    \usepackage{cancel} 
    \usepackage{enumitem} 
    \usepackage{comment}

    \usepackage{amsmath,amssymb,amsfonts,amsthm} 
    \usepackage{mathtools,todonotes} 
    \usepackage{tikz-cd}  
    \usepackage{tikz}
        \usetikzlibrary{arrows}
    \usepackage[all,cmtip]{xy} 
    \usepackage[bbgreekl]{mathbbol} 
    \usepackage{array} 

    \usepackage{hyperref} 
    
\theoremstyle{definition} 
    \newtheorem{definition}{Definition}

\theoremstyle{plain} 
    \newtheorem{theorem}[definition]{Theorem}
    \newtheorem{proposition}[definition]{Proposition}
    \newtheorem{lemma}[definition]{Lemma}
    \newtheorem{corollary}[definition]{Corollary}

\theoremstyle{remark} 
    \newtheorem{remark}[definition]{Remark}

 \usepackage[bibstyle=alphabetic,citestyle=alphabetic,useprefix,giveninits=true, sorting=ynt, sortcites, minbibnames=99,maxbibnames=99,backend=biber]{biblatex}  
  \renewbibmacro{in:}{} 
  \bibliography{bibliography,Additional_bibliography}  
 \DeclareSourcemap{
   \maps[datatype=bibtex]{
     \map[overwrite]{ 
       \step[fieldsource=doi, final]
       \step[fieldset=url, null]
       \step[fieldset=eprint, null]
       }
      \map[overwrite]{ 
       \step[fieldsource=eprint, final]
       \step[fieldset=pages, null]
       \step[fieldset=eid, null]
       \step[fieldset=journal, null]
     }  
   }
 }

     \newcommand{\tc}{\widetilde{c}}
     \newcommand{\tom}{\widetilde{\omega}}
     \newcommand{\tv}{\widetilde{v}}
     \newcommand{\te}{\widetilde{e}}
     \newcommand{\txi}[1]{\widetilde{\xi}^{#1}}
    \newcommand{\whom}{\widehat{\omega}}

\newcommand{\be}{\mathbf{e}}
\newcommand{\bom}{\boldsymbol{\omega}}

\newcommand{\bxi}{\boldsymbol{\xi}}
\newcommand{\bc}{\mathbf{c}}

\newcommand{\calO}{{\mathcal O}}

\newcommand{\AKSZ}{\textsf{\tiny AKSZ}}

    \makeatletter
        \newcommand{\zzlabel}[1]{\ifmeasuring@\else\ltx@label{#1}\fi} 
    \makeatother

    \newcounter{terms}[section] 
    \newcommand{\unl}[2]{\underline{#1}_{\refstepcounter{terms} \zzlabel{#2} \theterms}} 

    \newcommand{\reft}[2]{(\ref{#1}.\ref{#2})} 

\title{BV Pushforward of Palatini--Cartan gravity}
\author{Giovanni Canepa$^\alpha$, Alberto S. Cattaneo$^\beta$\thanks{ASC acknowledges partial support of the SNF Grant No.\ 200020\_192080 and of the Simons Collaboration on Global Categorical Symmetries. This research was (partly) supported by the NCCR SwissMAP, funded by the Swiss National Science Foundation, and is based upon work from COST Action 21109 CaLISTA, supported by COST (European Cooperation in Science and Technology) (\href{www.cost.eu}{www.cost.eu}), MSCA-2021-SE-01-101086123 CaLIGOLA, and MSCA-DN CaLiForNIA - 101119552. GC acknowledges partial support of the SNF Grants No.\ P5R5PT\_222221 and No.\ P500PT\_203085. This material is partially based upon work supported by the Swedish Research Council under grant no. 2021-06594 while the authors were in residence at Institut Mittag-Leffler in Djursholm, Sweden during the first semester of 2025.}}

\begin{document}

\maketitle

\begin{center}
$\prescript{\alpha}{}{}$INFN, Sezione di Firenze \\ Via Sansone 1, 50019 Sesto Fiorentino (FI), Italy \\
\end{center}
\begin{center}
$\prescript{\beta}{}{}$Universität Zürich, Institut für Mathematik, \\Winterthurerstrasse 190, 8057 Zürich (ZH), Switzerland\\
\end{center}

\begin{abstract}
    The goal of this note is to show that the standard BV formulation of gravity in the Palatini--Cartan formalism is equivalent, using a BV pushforward, to an
    AKSZ-like version compatible with BFV data on the boundary.
\end{abstract}

\tableofcontents

\section{Introduction}
Gravity in the Palatini--Cartan (PC) formalism\footnote{that is, working with Cartan's coframes instead of metric tensors and letting the connection be an independent field, \`a la Palatini, to be determined by the equations of motion \cite{Cartan, Palatini1919}} has the advantage of using differential forms, which makes the boundary (and lower-dimensional corners) analysis more transparent. On the flip side, in dimension $4$ and higher, its BV formulation \cite{BV1} is not directly compatible with the boundary BFV formulation \cite{BV3}. As such, PC gravity cannot be directly presented as an extended field theory, which would be the first, semiclassical, step towards a quantization satisfying Segal's axioms \cite{Segal1988}.

 A possible way out, on a ``globally hyperbolic'' space--time $M=\Sigma\times I$, $\Sigma$ being a Cauchy surface, was studied in \cite{CCS2020b} (see also \cite{PhysRevD.58.124029,Alexandrov_2000,BarrosESa2001,Date:2008rb,Rezende:2009sv,Bodendorfer_2013, CS2019, CCS2020,  MERC:2019, MRC:2019}), where the AKSZ construction \cite{AKSZ} was applied to the BFV description of the boundary's reduced phase space. The resulting theory automatically has a compatible BV-BFV structure \cite{CMR2012} and is BV embedded in the BV description of PC gravity. The goal of this paper is to show that this embedding is actually a BV equivalence, via BV pushforward.

More explicitly, we consider PC theory in $N$ space--time dimensions, $N\ge4$. We pick once and for all a vector bundle $\mathcal V$ over space--time $M$ isomorphic to its tangent bundle and carrying a Minkowski (or Euclidean) inner product.\footnote{{More formally $\mathcal{V}$ is the associated vector field of a $SO(N-1,1)$-principal bundle with typical fiber $(V,\eta)$, an $n$-dimensional vector space with a Minkowski (or Euclidean) inner product $\eta$ on it.}}
The action reads
    \[
        S[\be,\bom]=\int_M \frac{1}{(N-2)!}	\be^{N-2} F_{\bom}, 
    \]
where $\be\colon TM\xrightarrow\sim{\mathcal V}$ is a coframe and $\bom$ is an orthogonal connection for $\mathcal V$. The BV action for this theory was studied in \cite{P00,CS2017} and is recalled in Section~\ref{s:sBV}. 

If $M$ has a boundary $\Sigma$, we require $\be$ to induce a nondegenerate metric on it (we might also more strictly require the metric induced by $\be$ to be positive-definite, which would make $\Sigma$ space-like). Under this assumption, it was shown in \cite{CS2019} that the reduced phase space can be obtained as the reduction of a coisotropic submanifold (i.e., defined by first-class constraints) of a symplecting space of boundary fields consisting $\mathcal V$-coframes and connections on $\Sigma$ satisfying an appropriate condition (called the structural constraint). As such, the reduced phase space can be cohomologically resolved in terms of the BFV formalism. This is explicitly done in \cite{CCS2020}.

The problem is that the bulk BV theory does not communicate well with the boundary BV theory, due to some irregularity in a quotient appearing in the course of the construction \cite{CS2017}. 

In this paper, also extending the classical construction of \cite{CC2025}, we show how to remedy this when working on a space--time $M$ of the form $\Sigma\times I$, where $I$ is an interval or a circle, and $\be$ is required to induce a nondegenerate metric on each slice $\Sigma\times\{t\}$, $t\in I$. (Under the stricter condition that each of these metrics is positive-definite, this makes $M$ into a globally hyperbolic manifold for each allowed coframe field $\be$.) We can now split the BV fields of the theory into their $\Sigma$ and $I$ components. It turns out that part of the $\Sigma$  component of the connection is redundant, as already observed at the classical level in \cite{CC2025}. Here we show that upon gauge fixing it can be integrated away. This is a partial BV integration, known as a BV pushforward, which in the present case leads to an equivalence between the original theory and the one resulting after pushforward. The latter is then shown to be that described in \cite{CCS2020b} as a result of the AKSZ construction for the BFV structure associated to $\Sigma$. In conclusion, we prove that PC gravity, on $\Sigma\times I$, is equivalent to a theory with a well-defined BV-BFV structure, in the sense specified in Section \ref{s:observables}. In particular, this means that the expectation values of corresponding observables in the two theories are the same, hence carrying the same physical information.

{We can summarize the content of this paper with the following diagram which shows the connection between the standard BV theory $\mathfrak{F}_s$ and the one coming from the AKSZ construction $\mathfrak{F}_{\AKSZ}$. This diagram holds for cylindrical manifolds, and for each arrow we specify the Theorems or the corresponding reference.}

\begin{center}
\usetikzlibrary{decorations.pathmorphing}
\begin{tikzcd}[row sep= 4 em, column sep= 8 em]
    \mathsf{F} 
        \arrow[r, dashed, "BV", "\text{\cite{CS2017,P00}}"']
        \arrow[dd, dashed, "\text{KiTu}"',"\text{\cite{CCS2020}}"]
    & \mathfrak{F}_s 
        \arrow[dd, "\cancel{BV-BFV}"',"\text{\cite{CS2017}}", dashed, "\times\times\times" marking]
        \arrow[r, leftrightarrow,"\text{strong BV equiv.}", "\text{Theorem \ref{t:split_sympl_form}}"']
    & \mathfrak{F}_H   
        \arrow[d, "\text{BV pushforward}"', "\text{Theorem \ref{t:BV_pushforward}}"]
    & \\
        && \mathfrak{F_r} \arrow[r, leftrightarrow, "\text{strong BV equiv.}", "\text{\cite{CCS2020b}}"']
        & \mathfrak{F}_{\AKSZ}
         \\
    \mathsf{F}^{\partial}   
        \arrow[r, dashed, "\beta", "\text{\cite{CCS2020}}"'] 
    &\mathfrak{F}_{BFV} 
        \arrow[rru, leftrightarrow, "\text{AKSZ \cite{CCS2020b}}","\text{BV-BFV}"', bend right=20]
    & 
    & 
\end{tikzcd}
\end{center}

{Note that $\mathsf{F}$ and $\mathsf{F}^{\partial}$ are the classical theory on the bulk and the reduced phase space on the boundary, connected through the Kijowski-Tulcijev construction and $\beta$ represents the construction of a BFV theory $\mathfrak{F}_{BFV}$ out of a coisotropic submanifold describing the reduced phase space of the theory. Finally $\mathfrak{F}_{\AKSZ}$ is defined in \cite{CCS2020b}, while $\mathfrak{F}_s$, $\mathfrak{F}_r$ and $\mathfrak{F}_H$ are defined in Definition \ref{def:standardPC-BV}, Definition \ref{def:restrictedPC-BV} and Theorem \ref{t:split_sympl_form} respectively.}

\subsection*{Summary of the paper}
In Section~\ref{sGBVpf}, we recall the BV pushforward and its properties. 
In Section~\ref{s:twoBV}, we recall the BV formalism for PC theory and for the AKSZ construction starting from BFV on a boundary. In Sections~\ref{s:symplectomorphism-BV-to-Hedgehog} and~\ref{s:hedge}, we compute the BV pushforward showing the equivalence between the two theories, ignoring boundary effects. This is legitimate if $I$ is a circle or if $I$ is an open interval and we let the fields decay to zero towards its boundary. If $I$ is a compact interval, which is the interesting case to relate bulk to boundary, we need a more refined analysis, taking care of boundary conditions. This is explained in Section~\ref{s:boundary}.

{\subsection*{Acknowledgements}
We thank the anonymous referees for very useful comments that led to an improvement of the text and Pavel Mnev for engaging and insightful discussions.
}

\section{Generalities on the BV pushforward}\label{sGBVpf}
A short introduction to the BV-pushforward can be found in \cite{CR2005, CM2009,M2008}, in \cite{Mnev2017}[Section 4.7] and in \cite{CMR2}[Section 2.2.2]. 

{
\begin{remark}\label{r:infdimcase}
    In this section, we focus on a presentation of the finite-dimensional case, where no mathematical issues arise. The extension to the infinite-dimensional case is done in the setting of perturbative expansion around a Gaussian in terms of Feynman diagrams. This, in general, requires regularization and renormalization both of the Feynman diagrams and of the BV Laplacian $\Delta$ (a good introduction to this is \cite{Cos11book}). In the particular case used in this paper, see Section~\ref{s:backtoPC}, we are in the very simple situation where the Gaussian integration only requires a minimal regularization (at each point in spacetime we only have a finite-dimensional integral) and the perturbation, bilinear in fields and antifields, vanishes for the chosen gauge fixing.
\end{remark}
}

Let $(\mathcal{F},\varpi)$ be a graded symplectic manifold with symplectic form of degree $-1$. Let us suppose that there exist two graded symplectic manifolds with symplectic forms of degree $-1$, $(\mathcal{F}_1,\varpi_1)$ and $(\mathcal{F}_2,\varpi_2)$ such that \begin{align*}
    \mathcal{F}&= \mathcal{F}_1 \times \mathcal{F}_2 & \varpi &= \varpi_1 + \varpi_2.
\end{align*}

\begin{definition}\label{d:BV-pushforward}
    Given a Lagrangian submanifold $\mathcal{L}\subset \mathcal{F}_2$, the BV-pushforward along $\mathcal{L}$ is a partial\footnote{Its definition domain consists of half-densities integrable on $\mathcal L$.} map between half-densities,
\begin{align*}
    \mathcal{P}_{\mathcal{L}} : \mathrm{Dens}^{\frac{1}{2}}\mathcal{F} \rightarrow \mathrm{Dens}^{\frac{1}{2}}\mathcal{F}_1
\end{align*}
  defined as an $\mathbb{R}$-linear map sending $\phi = \phi_1 \otimes \phi_2 \in \mathrm{Dens}^{\frac{1}{2}}\mathcal{F}_1\otimes \mathrm{Dens}^{\frac{1}{2}}\mathcal{F}_2 \subset \mathrm{Dens}^{\frac{1}{2}}\mathcal{F}$ to 
\begin{align*}
    \mathcal{P}_{\mathcal{L}} \phi : =  \phi_1 \cdot \int_{\mathcal{L}} \phi_2
\end{align*}
where the integral is the  BV-integral over $\mathcal{L}$.
\end{definition}

The importance of this map is related to the fact that it relates solutions of the Quantum Master Equation (QME) on $\mathcal{F}$ to solutions of the QME on $\mathcal{F}_1$. Let us now briefly recap the results that lead to this statement.
The first is an adaptation of the standard BV theorem.
\begin{theorem}[\cite{M2008}] \label{t:BV-pushforward_properties}
Let $\Delta, \Delta_1$ and $\Delta_2$ be the canonical BV laplacians on respectively $\mathcal{F}, \mathcal{F}_1$ and $\mathcal{F}_2$.
    \begin{enumerate}
        \item Given any Lagrangian submanifold $\mathcal{L}\subset \mathcal{F}_2$, $\mathcal{P}_{\mathcal{L}}$ is a chain map, i.e.
        \begin{align*}
            \Delta_1 \mathcal{P}_{\mathcal{L}} = \mathcal{P}_{\mathcal{L}} \Delta.
        \end{align*}
        \item If $\mathcal{L}\sim \mathcal{L}'$ are two homologous Lagrangians in $\mathcal{F}_2$ and $\phi \in \mathrm{Dens}^{\frac{1}{2}}\mathcal{F}$ is such that $\Delta\phi=0$, then the difference 
        \begin{align*}
            \mathcal{P}_{\mathcal{L}}\phi - \mathcal{P}_{\mathcal{L}'}\phi
        \end{align*}
        is $\Delta_1$-exact.
    \end{enumerate}
\end{theorem}

Let us now assume that the symplectic manifolds $(\mathcal{F},\varpi), (\mathcal{F}_1,\varpi_1)$ and $(\mathcal{F}_2,\varpi_2)$ are equipped with compatible Berezinians $\mu, \mu_1$ and $\mu_2$ and such that $\mu= \mu_1 \mu_2$. {Define $\Delta_{\mu}(f):=\frac{1}{2}\mathrm{div}_{\mu}X_f$ where $X_f$ is the Hamiltonian vector field of $f$, and note that $\sqrt{\mu}\Delta_{\mu}(f)= \Delta(\sqrt{\mu}f).$} Then we have the following corollary.

\begin{corollary}\label{c:BV-pushforward_QME}
    Let $S \in C^{\infty}(\mathcal{F})[[\hbar]]$. Then there exists $S_1\in C^{\infty}(\mathcal{F}_1)[[\hbar]]$ such that
    \begin{align*}
        \mu_1^{\frac{1}{2}} e^{\frac{i}{\hbar} S_1} = \mathcal{P}_{\mathcal{L}} \left( \mu^{\frac{1}{2}} e^{\frac{i}{\hbar} S}\right).
    \end{align*}
     If $S$ satisfies the Quantum Master Equation on $\mathcal{F}$
    \begin{align*}
        \Delta_{\mu}\left(e^{\frac{i}{\hbar} S}\right)=0,
    \end{align*}
    then so does $S_1$ on $\mathcal{F}_1$:
    \begin{align*}
        \Delta_{\mu_1}\left(e^{\frac{i}{\hbar} S_1}\right)=0.
    \end{align*}
    Furthermore, if we have $S_1$ and $S_1'$ for two homologous Lagrangians $\mathcal{L}\sim \mathcal{L}'\subset \mathcal{F}_2$,
    then 
    \begin{align*}
        e^{\frac{i}{\hbar} S_1} -e^{\frac{i}{\hbar} S_1'}
    \end{align*}
    is $\Delta_1$-exact.
\end{corollary}

\subsection{BV bundles {and BV products}}\label{s:BV_bundles}
The setting explained at the beginning of this section can (and must) be extended, with some care, to the case of a fiber bundle $\pi\colon\mathcal{F}\to \mathcal{F}_1$ that is locally the product of odd symplectic manifolds. By choosing a lagrangian submanifold $\mathcal{L}$ of the fibers (we can in principle allow a smoothly varying family of Lagrangian submanifolds $\mathcal{L}_x\subset\pi^{-1}(x)$, $\forall x\in\mathcal{F}_1$), we can get, via integration along $\mathcal{L}$, a map
\[
\mathcal{P}_{\mathcal{L}} \colon \mathrm{Dens}^{\frac{1}{2}}\mathcal{F} \rightarrow \mathrm{Dens}^{\frac{1}{2}}\mathcal{F}_1.
\]
What is not guaranteed in the general setting are the results of Theorem~\ref{t:BV-pushforward_properties}.

Conditions on the fiber bundle $\pi\colon\mathcal{F}\to \mathcal{F}_1$ to obtain Theorem~\ref{t:BV-pushforward_properties} are discussed in \cite[Remark 2.13]{CMR2} under the name of a ``BV hedgehog.''

For the present note, we only need a simpler version, {called a BV product,} which we explain now. Namely, assume that we have a bundle isomorphism
\begin{center}
\begin{tikzcd}[row sep= 4 em, column sep= 8 em]
    \mathcal{F} \arrow[r, "\Phi"]\arrow[d, "\pi"']& \widetilde{\mathcal{F}} \arrow[d, "\widetilde{\pi}"]\\
\mathcal{F}_1 \arrow[r, "\phi"'] & \widetilde{\mathcal{F}_1}
\end{tikzcd}
\end{center}
where $ \widetilde{\mathcal{F}}=\widetilde{\mathcal{F}_1}\times\widetilde{\mathcal{F}_2}$, $\widetilde\pi$ is the projection to the first factor, and both $\Phi$ and $\phi$ are symplectomorphisms. 
We then choose a Lagrangian submanifold $\widetilde{\mathcal{L}}$ of $\widetilde{\mathcal{F}_2}$ and perform the BV pushforward along it. What we are interested in is the composition
\[
\widehat{\mathcal{P}}_{\widetilde{\mathcal{L}}} \coloneqq (\phi^{-1})_*\circ \mathcal{P}_{\widetilde{\mathcal{L}}} \circ \Phi_*.
\]
Since the pushforwards of half-densities by symplectomorphisms are chain maps, the first statement of Theorem~\ref{t:BV-pushforward_properties} now holds for $\widehat{\mathcal{P}}_{\widetilde{\mathcal{L}}}$. The second statement holds as well for  homologous Lagrangians $\widetilde{\mathcal{L}}\sim \widetilde{\mathcal{L}'}$ in $\widetilde{\mathcal{F}_2}$. Similarly, we get the results of Corollary~\ref{c:BV-pushforward_QME}.

\subsection{BV observables and quasiisomorphisms} \label{s:observables}
Two theories connected through a BV pushforward can be considered equivalent in the sense specified in this section.
Suppose $S$ satisfies the QME on $\mathcal{F}$ with respect to a reference, $\Delta$-closed, nowhere vanishing half-density $\mu^{\frac12}$ (i.e., $\Delta_{\mu}\left(e^{\frac{i}{\hbar} S}\right)=0$). Thanks to the QME, the degree-$1$ operator
\[
\Omega_{(S,\mu)} \coloneqq (S,\ ) -i\hbar \Delta_\mu
\]
acts on a functional $\calO$ as
\[
\Omega_{(S,\mu)}\calO = -i\hbar\, e^{-\frac{i}{\hbar} S}\Delta_\mu \left( e^{\frac{i}{\hbar} S}\calO\right).
\]
Therefore, it squares to zero and defines a cohomology $H_{\Omega_{(S,\mu)}}(\mathcal{F})$.
Note that
\begin{enumerate}
\item $\Omega_{(S,\mu)}\calO = 0$ if and only if $\Delta \left(\mu^{\frac12} e^{\frac{i}{\hbar} S}\calO\right)=0$.
\item $\calO=\Omega_{(S,\mu)}\widetilde\calO$ if and only if $\mu^{\frac12}e^{\frac{i}{\hbar} S}\calO=\Delta \left(-i\hbar\, \mu^{\frac12}e^{\frac{i}{\hbar} S}\widetilde \calO\right)$.
\end{enumerate}
This implies that the expectation value
\[
\langle \calO \rangle_{\mathcal L} \coloneqq \frac{\int_{\mathcal L}\mu^{\frac12}e^{\frac{i}{\hbar} S}\calO}{{\int_{\mathcal L}\mu^{\frac12}e^{\frac{i}{\hbar} S}}}
\]
is invariant under deformations of $\mathcal L$ if $\calO$ is $\Omega_{(S,\mu)}$-closed and vanishes if $\calO$ is $\Omega_{(S,\mu)}$-exact. Therefore, the expectation value is  a well-defined partial linear functional on the cohomology $H_{\Omega_{(S,\mu)}}(\mathcal{F})$, whose elements are accordingly called BV observables.

A BV symplectomorphism that maps $(S,\mu)$ to $(\widetilde S,\widetilde \mu)$ establishes an isomorphic chain map between the corresponding complexes and therefore an isomorphism in cohomology.

Another way to produce a chain map is via the BV pushforward. Namely, with the notations of Corollary~\ref{c:BV-pushforward_QME}, we have the partial chain map\footnote{This map depends on the choice of $\mu$, $S$, and ${\mathcal{L}}$, which we do not write down explicitly in order to keep the notation light.} 
\[
p\colon
\begin{array}[t]{ccc}
 (\mathcal{F}, \Omega_{(S,\mu)}) & \to & (\mathcal{F}_1, \Omega_{(S_1,\mu_1)})\\
 \calO & \mapsto & \mu_1^{-\frac12}e^{-\frac{i}{\hbar} S_1} \mathcal{P}_{\mathcal{L}} \left( \mu^{\frac12}e^{\frac{i}{\hbar} S}\calO\right)
 \end{array}
\]
which in turn induces a linear map $H_{\Omega_{(S,\mu)}}(\mathcal{F})\to H_{\Omega_{(S_1,\mu_1)}}(\mathcal{F}_1)$. 
Note that we have
\[
\int_{{\mathcal L_1}\times{\mathcal L}}\mu^{\frac12}e^{\frac{i}{\hbar} S}\calO = 
\int_{{\mathcal L}_1}\mu_1^{\frac12}e^{\frac{i}{\hbar} S_1}p\calO,
\]
so\footnote{Note that $\int_{{\mathcal L_1}\times{\mathcal L}}\mu^{\frac12}e^{\frac{i}{\hbar} S} = 
\int_{{\mathcal L}_1}\mu_1^{\frac12}e^{\frac{i}{\hbar} S_1}$, so the normalization factors are the same on both sides.}
\[
\langle \calO \rangle_{{\mathcal L_1}\times{\mathcal L}} = 
{\langle p\calO  \rangle_{{\mathcal L}_1}}.
\]

When the map induced by $p$ in cohomology
 is an isomorphism, the chain map $p$ is called a quasiisomorphism. This is equivalent to having a partial chain map $q\colon (\mathcal{F}_1, \Omega_{(S_1,\mu_1)}) \to (\mathcal{F}, \Omega_{(S,\mu)})$ such that $q\circ p$ and $p\circ q$ are homotopic to the identity maps. It follows that, for a functional $\calO_1$ on $\mathcal{F}_1$, we have
 \[
 \langle \calO_1  \rangle_{{\mathcal L}_1}
 ={\langle q\calO_1 \rangle_{{\mathcal L_1}\times{\mathcal L}}}.
 \]

Therefore, when a BV pushforward defines a quasiisomorphism, the two theories are completely equivalent in the sense that there is a one-to-one correspondence between their BV observables and between the respective expectation values.

We will show that the two formulations of gravity considered in this paper are equivalent in this sense {(Theorems \ref{t:BV_pushforward} and \ref{t:BV-pushforward_boundary}).}

\section{Two BV theories of PC gravity on cylinders}\label{s:twoBV}
\subsection{\emph{Standard} BV theory}\label{s:sBV}
In this section we introduce the BV version of Palatini--Cartan theory. This description firstly appeared in {\cite{P00} in the BRST formulation and later generalised to BV in }\cite{CS2017} . Using the notation of Section \ref{sGBVpf}, this theory will play the role of $\mathcal{F}$ in the BV pushforward.

Let $M$ be an $N$-dimensional manifold with $N>2$  {and let $P\rightarrow M$ be a principal $SO(N-1,1)$-bundle with $\mathcal{V}$ its associated bundle with fiber an $N$-dimensional Minkowski vector space. We denote by $\Omega^i(M, \wedge^j\mathcal{V})$ the space of $i$-forms with value in $\wedge^j\mathcal{V}$, and with $\mathcal{A}(M)$ the $P$ principal connections on $M$. The subscript \emph{nd} means that the form is non-degenerate in the sense explained in the introduction.}

\begin{definition}\label{def:standardPC-BV}
    The \emph{standard} BV theory for Palatini--Cartan gravity is the quadruple 
    $$\mathfrak{F}_{s}(M)= (\mathcal{F}_{s}(M), S_{s}(M), \varpi_{s}(M), Q_{s}(M)),$$
    where the quantities appearing above are defined as follows:
    \begin{enumerate}
    \item $\mathcal{F}_{s}(M)$ is the space of fields and it is defined as 
        \begin{equation*}
        \mathcal{F}_{s}(M)\coloneqq T^*[-1]\left(\Omega_{nd}^1(M, \mathcal{V}) \oplus\mathcal{A}(M) \oplus \mathfrak{X}[1](M) \oplus \Omega^0[1](M,{\wedge^2\mathcal{V}})\right),
        \end{equation*} 
        where the fields in the base are denoted by $(\be, \bom, \bxi^{}, \bc)$, while the corresponding variables in the cotangent fibre are denoted by $(\be^{\dag}, \bom^{\dag}, \bxi^{\dag}, \bc^{\dag})$.
    \item $\varpi_{s}(M)$ is a symplectic form on $\mathcal{F}_{s}(M)$ and its explicit expression is 
        \begin{equation*}
        \varpi_{s}(M) = \int_M \delta \be \delta \be^{\dag} + \delta \bom \delta \bom^{\dag}+ \delta \bc \delta \bc^{\dag} + \iota_{\delta \bxi^{}} \delta\bxi^{\dag}.
        \end{equation*}
    \item $S_{s}(M) \in C^{\infty}(\mathcal{F}_{s}(M))$ is the BV action functional and reads
        \begin{align}\label{e:standard_action}
        S_{s}(M) &=\int_M \frac{1}{(N-2)!}	\be^{N-2} F_{\bom}  + \left(\iota_{\bxi^{}} F_{\bom}  - d_{\bom} \bc \right)\bom^\dag - \left(L_{\bxi^{}}^{\bom}\be- [\bc,\be]\right)\be^\dag\\
        &\phantom{=}+\int_M \frac{1}{2}\left(\iota_{\bxi^{}}\iota_{\bxi^{}} F_{\bom}  - [\bc,\bc]\right)\bc^\dag +\frac12\iota_{[\bxi^{},\bxi^{}]}\bxi^{\dag}, \nonumber
        \end{align}
        where $L_{\bxi^{}}^{\bom}\be = \iota_{\bxi^{}}d_{\bom}\be- d_{\bom}\iota_{\bxi^{}}\be$.
    \item The cohomological vector field $Q_{s}$ is defined by the equation $\iota_{Q_{s}}\varpi_{s}= \delta S_{s}$ and its action on components is 
        \begin{align*}
        Q_{s} \be &=  L_{\bxi^{}}^{\bom}\be- [\bc,\be] & 
        Q_{s} \be^{\dag} & = \be^{N-3} F_{\bom} + L_{\bxi^{}}^{\bom} \be^{\dag} - [\bc , \be^{\dag}]\\
        Q_{s} \bom &= \iota_{\bxi^{}} F_{\bom}  - d_{\bom} \bc &
        Q_{s} \bom^{\dag} &= \be^{N-3} d_{\bom} \be - d_{\bom} \iota_{\bxi^{}} \bom^{\dag} - [\bc, \bom^{\dag}] + \iota_{\bxi^{}}[\be, \be^{\dag}]-\frac{1}{2} d_{\bom} \iota_{\bxi^{}} \iota_{\bxi^{}} \bc^{\dag}\\
        Q_{s} \bc &= \frac{1}{2}\iota_{\bxi^{}}\iota_{\bxi^{}} F_{\bom}  - \frac{1}{2}[\bc,\bc] &
        Q_{s} \bc^{\dag} &= - d_{\bom} \bom^\dag - [\be, \be^{\dag}] - [\bc, \bc^{\dag}] \\
        Q_{s} \bxi^{} &= \frac{1}{2} [\bxi^{}, \bxi^{}] &
        Q_{s} \bxi^{\dag}_{\bullet} &= F_{\bom\bullet} \bom^\dag -  (d_{\bom\bullet}\be)\be^{\dag}+\iota_{\bxi^{}}F_{\bom\bullet} \bc^\dag + L_{\bxi^{}}^{\bom}\bxi^{\dag}_{\bullet} + (d_{\bom}\iota_{\bxi^{}} \bxi^{\dag})_{\bullet}.
        \end{align*}
        Here we used the symbol $\bullet$ to remind the reader that $\bxi^\dag$ is a one-form with values in densities, and on the right hand side we highlight the one-form part of the expression.
    \end{enumerate}
\end{definition}

{
Suppose now that $M = \Sigma \times I$ where $\Sigma$ is an $(N-1)$-dimensional closed manifold and $I$ is a 1-dimensional manifold.
The product structure of $M$ induces a splitting on forms and on vector fields on $M$: let $\boldsymbol{\phi} \in \Omega^k(\Sigma \times I)$, then we get
\begin{align*}
    \boldsymbol{\phi} = {\phi} + \underline{{\phi}_n}
\end{align*}
where ${\phi} \in \Omega^k(\Sigma)\widehat\otimes C^{\infty}(I)$ and $\underline{{\phi}_n} \in \Omega^{k-1}(\Sigma)\widehat\otimes \Omega^{1}(I)$ {where the underline is a bookeeper for $ \Omega^{1}(I)$.\footnote{Note that if $k=n$ the first addend does not exists, while for $k=0$ only the first survives.}} On the other hand, for a vector field $v \in \mathfrak{X}(M)$ we obtain
$\mathbf{v}= {v} + \overline{{v}^n}$ where ${v} \in \mathfrak{X}(\Sigma)\widehat\otimes C^{\infty}(I)$ and  $\overline{{v}^n} \in \mathfrak{X}(I)\widehat\otimes C^{\infty}(\Sigma)$. {Starting from a field $\underline{{\phi}_n}$ we can then define a not-underlined version via $\underline{{\phi}_n}= {\phi}_n dx^n$. This notation will be useful when contracting a form $\boldsymbol{\phi}$ with a vector field $\mathbf{v}$ as
$\iota_{\mathbf{v}}\boldsymbol{\phi} = \iota_v \phi + \iota_v \underline{{\phi}_n} + {v}^n {\phi}_n$ from which we can deduce visually which terms are one-forms on $I$.}
\begin{remark}
    Note that we underline  forms that are one-forms along $I$. In particular, $\underline{{\phi}_n} = {\phi}_n dx^n$
    where $x^n$ is the coordinate along $I$. {This also helps in order to keep track of the parity in the computations, 
    the parity of $\underline{{\phi}_n}$ being} the same as $\boldsymbol{\phi}$, while ${\phi}_n$ has opposite parity. 
\end{remark}

The fields of the BV formulation of Palatini--Cartan theory then get split as 
\begin{align*}
    \be &= e + \underline{e_n} &  
    \be^\dag &= e^\dag + \underline{e^{\dag}_n} &
    \bom &= \omega + \underline{\omega_n} &  
    \bom^\dag &= \omega{}^\dag + \underline{\omega^{\dag}_n}\\
    \bc &= c & 
    \bc^\dag &= \underline{c^\dag_n} &
    \bxi &= \xi + \overline{{\xi^{n}}} & 
    {\bxi}^\dag &= \underline{\xi^{\dag}} + \underline{{\xi^{\dag}}_n}
\end{align*}

Hence, on a cylindrical manifold $\Sigma \times I$, we can write the symplectic form $\varpi_s$ and $S_s$ as follows:
\begin{align}\label{e:standard_symplectic_form_split}
    \varpi_s = \int_{\Sigma \times I}\delta e \delta \underline{e_n^\dag} + \delta \underline{e_n} \delta e^\dag + \delta \omega \delta \underline{\omega_n}^\dag + \delta \underline{\omega_n} \delta \omega^\dag  + \delta c \delta \underline{c_n^\dag} + \iota_{\delta \xi} \delta \underline{\xi^{\dag}} + \delta \xi^{n} \delta \underline{\xi^{\dag}_n}.
\end{align}

\begin{align}\label{e:standard_action_split_beforev}
    S_{s}& =  \int_{\Sigma \times I}\frac{1}{(N-3)!}  e^{N-3} \underline{e}_n F_{\omega} + \frac{1}{(N-2)!} e^{N-2} \underline{F_{\omega_n}}-  \left(\iota_{\xi{}} F_{\omega}+ F_{\omega_n}\xi^n  - d_{\omega} c \right)\underline{\omega_n^\dag} \\
    \nonumber&\phantom{=  \int_{\Sigma \times I}} - \left(\iota_{\xi{}} \underline{F_{\omega_n}}  - \underline{d_{\omega_n}} c \right)\omega^{\dag} + \left(L_{\xi{}}^{\omega }e  + d_{\omega_n} e \xi^n  + e_n d \xi^n - [c,e] \right)\underline{e}_n^{\dag}\\ 
    \nonumber&\phantom{=  \int_{\Sigma \times I}}+ \left(\iota_{\xi{}}d_{\omega}\underline{e}_n  + \iota_{\underline{\partial}_n \xi{}}e  - \underline{d_{\omega_n}}(e_n \xi^n) - [c,\underline{e}_n] \right)e^\dag\\
    \nonumber&\phantom{=  \int_{\Sigma \times I}}- \left(\frac{1}{2}\iota_{\xi{}}\iota_{\xi{}} F_{\omega}    + \iota_{\xi{}} F_{\omega_n}\xi^n  - \frac{1}{2}[c,c]\right)\underline{c}_n^\dag +\frac12 \iota_{[\xi{},\xi{}]}\underline{\xi^{\dag}}+ \xi^n\iota_{\partial_n\xi{}}\underline{\xi^{\dag}}+  \iota_{\xi{}}d\xi^n\underline{\xi^{\dag}}_n+ \xi^n\partial_n\xi^n\underline{\xi^{\dag}}_n,
\end{align}
where we defined $\underline{F_{\omega_n}} = d_{\omega} \underline{\omega_n} + \underline{\partial_n} \omega$.

\subsection{\emph{Restricted} BV theory}

Using the field $e$ we can then define the following set of maps
\begin{align*}
{W}_{e^{N-3}}^{i,j}\colon \Omega^i\left(M,\textstyle{\bigwedge^j}\mathcal V\right) &\rightarrow \Omega^{i+N-3}\left(M,\textstyle{\bigwedge^{j+N-3}}\mathcal V\right)\\
 Y &\mapsto e^{N-3} \wedge Y.
\end{align*}
The properties of these maps are collected in \cite{C2024}.

We now fix a nowhere vanishing section $\epsilon_n \in \Gamma(M,\mathcal{V})$  with the additional property that $d_I\epsilon_n=0$. Once $\epsilon_n$ is fixed, we restrict the space of fields by the open condition that $e$ and $\epsilon_n$
form a basis of $V$ at every point.  Since $\be$ is nondegenerate, $e_n$ becomes a linear combination of $e$ and $\epsilon_n$, with nonzero $\epsilon_n$-component. More generally we denote by $Y^{[\mu]}$ the components of a field $Y$ with respect to the basis given by $e$ and $\epsilon_n$ (i.e. $Y = Y^{[b]}e_b+Y^{[n]}{\epsilon}_n$). 
For shortness of notation we define{\footnote{We use here the underlined notation defined above.}}
\begin{align*}
   \underline{z}^a &: = \underline{{e}}_n^{[a]}   & \underline{\lambda} &:= \underline{{e}}_n^{[n]}
\end{align*}
so that
\begin{align*}
    \underline{e}_n = \iota_{\underline{z}} e + \underline{\lambda} \epsilon_n.
\end{align*}

Using this notation we consider $\underline{\mathfrak{W}^\dag} \in \Omega^{N-2}(\Sigma, \wedge^{N-2}\mathcal{V})\otimes \Omega^1(I)$ and ${X} \in \Omega^{0}(\Sigma \times I, \mathcal{V})$ defined as
\begin{align} \label{e:defWdag}
    \underline{\mathfrak{W}^\dag}&:= \underline{\omega}^{\dag}_n - \iota_{\underline{z}}\omega^\dag - \iota_{\xi}\underline{c}_n^{\dag}+\iota_{\underline{z}}{c}_{n}^{\dag}{\xi^{n}}\\
    X &:= [c, \epsilon_n]+ L_{\xi{}}^{\omega}(\epsilon_n)-d_{\omega_n} \epsilon_n\xi^{n}
\end{align}
and define the following BV theory.

\begin{definition}\label{def:restrictedPC-BV}
    The \emph{restricted} BV theory for Palatini--Cartan gravity is the quadruple 
    $$\mathfrak{F}_{r}= (\mathcal{F}_{r}, S_{r}, \varpi_{r}, Q_{r})$$
    where the quantities appearing above are defined as follows:
    \begin{enumerate}
    \item $\mathcal{F}_{r}$ is the space of fields and  is defined as the subspace of $\mathcal{F}_{s}$ of fields satisfying the following equations (called \emph{structural constraints}):
        \begin{align} \label{e:PCconstraints1}
            \underline{\mathfrak{W}}^\dag &\in \mathrm{Im}({W}_{e^{N-3}}^{1,1})\otimes \Omega^1(I)\\
            {\epsilon}_n e^{(N-4)} d_{\omega} e - {\epsilon}_n e^{(N-4)} {W}_{e^{N-3}}^{-1}&(\mathfrak{W}^\dag) d \xi^{n} + X^{[a]}(\omega_a^{\dag} -{c}_{an}^{\dag}{\xi^{n}}) \in \mathrm{Im}({W}_{e^{N-3}}^{1,1}), \label{e:PCconstraints2}
        \end{align}
    \item $\varpi_{r}=\varpi_{s}|_{\mathcal{F}_{r}}$;\footnote{Note that the structural constraints above restrict the space, hence the symplectic form is only \emph{written} in the same way. {The fact that $\varpi_{r}$ is still symplectic was proved in \cite{CCS2020b}.}}
    \item $S_{r}=S_{s}|_{\mathcal{F}_{r}}$;
    \item $Q_{r}=Q_{s}$.
    \end{enumerate}
\end{definition}

The importance of this theory was shown in \cite{CCS2020b} by proving that it is {strongly BV-equivalent to the AKSZ theory of gravity, i.e. a BV theory obtained from the reduced phase space of gravity via the BFV and AKSZ construction (see Definition \ref{def:strongBVeq}).

}

In this note we want to show that the \emph{restricted} BV theory described in Definition \ref{def:restrictedPC-BV} can be obtained as a BV pushforward of the \emph{standard} BV theory described in Definition \ref{def:standardPC-BV}. In particular we divide this in two steps. First we prove in Section \ref{s:symplectomorphism-BV-to-Hedgehog} that the \emph{standard} BV theory is symplectomorphic to a BV theory  isomorphic to a product of BV spaces. We then show that its BV pushforward does indeed coincide with the \emph{restricted} BV theory (Section \ref{s:BV-pushforward_Gravity_Hedgehog}).

\section{From the standard BV PC theory to a BV {Bundle}}\label{s:symplectomorphism-BV-to-Hedgehog}

The goal of this section is to show that, on cylindrical manifolds, the standard BV theory for Palatini--Cartan gravity $\mathfrak{F}_s$ introduced in Definition  \ref{def:standardPC-BV} is strongly BV equivalent (i.e. there exists a symplectomorphic compatible with the actions) to another BV theory $\mathfrak{F}_H$, in which the symplectic form can be written as the sum of the symplectic form of the restricted BV theory $\mathfrak{F}_r$ (see Definition \ref{def:restrictedPC-BV}) and a complement. Then, in Section \ref{s:gravity_hedgehog} we will show that {such theory $\mathfrak{F}_H$ is 
a BV product} with basis given by the restricted BV theory $\mathfrak{F}_r$.

{
Let us start with a definition of strong BV-equivalence:
\begin{definition}\label{def:strongBVeq}
A strong BV-equivalence between two BV theories $\mathfrak{F}_1$ and $\mathfrak{F}_2$ is a degree-$0$ symplectomorphism
$$ \Phi : (\mathcal{F}_1, \varpi_1 ) \rightarrow (\mathcal{F}_2, \varpi_2)$$
preserving the BV action\footnote{We always consider action functionals $S$ modulo constants.}: $\Phi^* S_2 = S_1$.
\end{definition}}

In order to construct the aforementioned symplectomorphism, we consider the following results:

{
    \begin{theorem}\cite{CS2019}[Theorem 33]\label{thm:omegadecomposition_d4}
Suppose that the  boundary metric $g^{\partial}$ is nondegenerate. Given any ${\omega} \in \Omega ( \Sigma,\wedge^2 \mathcal{V})$, there is a unique decomposition 
\begin{equation} \label{e:omegadecomp_d4}
{\omega}= \widehat{\omega} + v
\end{equation}
with $\widehat{\omega}$ satisfying  \eqref{e:PCconstraints2} and $v$ such that $e^{N-3} v=0 $.
\end{theorem}

\begin{lemma}\cite{CCS2020}[Lemma 16]\label{lem:Omega2,2_d4}
Let $g^\partial$ be nondegenerate and let $X \in \Omega^{N-2,N-2}_\partial$. Then there exist a unique $Z \in \ker W_{N-3}^{\partial, (1,2)}$ and a unique $Y \in \Omega_{\partial}^{1,1}$ such that 
\begin{align*}
\beta = e^{N-3} Y + \epsilon_n e^{N-4} [Z, e].
\end{align*}
\end{lemma}

The decomposition \eqref{e:omegadecomp_d4} allows us to write $\omega$ in terms of a field $\widehat{\omega}$ satisfying the constraint \eqref{e:PCconstraints2}, and hence being a field in $\mathcal{F}_{r}$, together with some residual field $v$.

Similarly Lemma \ref{lem:Omega2,2_d4} allows us to decompose $\underline{\mathfrak{W}^\dag}$ as
\begin{align*}
    \underline{\mathfrak{W}^\dag} = e^{N-3} \underline{\tau}^\dag + \epsilon_n e^{N-4} [\underline{\mu}^\dag, e]
\end{align*}
for some $\underline{\tau}^\dag$ and $\underline{\mu}^\dag$.
Using the explicit expression of $\underline{\mathfrak{W}^\dag}$ in \eqref{e:defWdag}, we can then define 
\begin{align*}
    \underline{\widehat{\omega}}^\dag_n &: = e^{N-3} \underline{\tau}^\dag + \iota_{\underline{z}}\omega^\dag+ \iota_{\xi{}}\underline{c}_n^{\dag}-\iota_{\underline{z}}{c}_{n}^{\dag}{\xi^{n}}\\
    \underline{v}^\dag &:=\epsilon_n e^{N-4} [\underline{\mu}^\dag, e],
\end{align*}
and rewrite $\underline{\omega}_n^\dag$ as
\begin{align}\label{e:split_omegan_dag}
    \underline{\omega}^{\dag}_n& = \underline{\whom}^\dag_n + \underline{v}^\dag.
\end{align}

}

{
\begin{theorem}\label{t:split_sympl_form}
    On a cylindrical manifold $\Sigma \times I$, there exists a {strong BV-equivalence} 
    \begin{align*}
        \psi\colon (\mathcal{F}_s,\varpi_s, S_s) \rightarrow (\mathcal{F}_H,\varpi_H, S_H),
    \end{align*}
    where $\mathcal{F}_s= \mathcal{F}_H$,
    \begin{align}\label{e:sympl_form_H}
        \varpi_H = \varpi_{r} + \int_{\Sigma \times I}\delta v \delta  \underline{v}^\dag,
    \end{align}
    and 
    \begin{align*}
        S_H = S_{r} + \int_{\Sigma \times I} \frac{1}{2(N-3)!}\underline{e_n}e^{N-3} [v,v] 
        +\left( L_{\xi}^{\omega} v + [c,v] + \xi^{n} d_{\omega_n} v  + d\xi^n \iota_z v \right)\underline{v}^\dag. 
    \end{align*}
\end{theorem}

In other words, the \emph{standard} BV theory is BV symplectomorphic to the \emph{restricted} one plus some additional parts proportional to $v$. We can then apply a BV pushforward to $(\mathcal{F}_H,\varpi_H, S_H)$, integrating $v$ in order to obtain the \emph{restricted} theory. This will be made explicit in Section \ref{s:BV-pushforward_Gravity_Hedgehog}.
}

\begin{proof}
\label{s:proofTheorem}
We prove the claim by explicitely showing a symplectomorphism $\psi$.

Let $\underline{v}^\dag$ as in \eqref{e:split_omegan_dag},  using Lemma \ref{lem:Omega2,2_d4} we can write $e^{N-4}\mu^\dag d\xi^n = e^{N-3} \alpha + \epsilon_n e^{N-4}[\beta, e]$. Since $\epsilon_n^2=0$, $\epsilon_n e^{N-4}\mu^\dag d\xi^n = \epsilon_n e^{N-3} \alpha= e^{N-3} \nu$ where we defined $\nu=\epsilon_n \alpha.$

Then the explicit expression of $\psi$ reads: 
    \begin{align*}
    \te &= e & 
    \underline{\te}_n^\dag &  = \underline{e}_n^\dag + \frac{1}{N-3} e^{N-4} \left( d_{\widehat{\omega}} (\epsilon_n \underline{\mu}^\dag) + \sigma \underline{\mu}^\dag+ \underline{\tau}^\dag (\nu+v)\right)\\
    &&&\phantom{=}+ X^{(a)}(Z_b \underline{\mu}^\dag)_a^{(b)} +
    \iota_{\underline{z}}((\nu_a+v_a) Z)^{(a)}+ \theta_N(5) e^{N-5}\epsilon_n \mu^\dag d\xi^n \underline{\tau}^\dag\\
    \underline{\te}_n & = \underline{e}_n & 
    \te^\dag & = e^\dag + ((\nu_a + v_a) Z)^{(a)} \\
    \widehat{\tom} & = \widehat{\omega} + \nu  &
    \underline{\tom}_n^\dag & =\underline{\omega}_n^\dag +  [\epsilon_n, \iota_{\xi}(Z_a \underline{\mu}^\dag)^{(a)}]\\
    \tv &= v & 
    \tom^\dag & = \omega^\dag  + [\epsilon_n, (Z_a \mu^\dag)^{(a)}\xi^n] \\
    \underline{\widetilde{v}}^\dag & = \underline{v}^\dag & 
    \underline{\tom}_n & = \underline{\omega}_n  + \iota_{\underline{z}}(\nu+ v) +X^{(a)}\underline{\mu}^\dag_a\\
    \txi{n} & = \xi^{n} &
    \tc & = c + \iota_{\xi}(\nu+v) + \iota_{z}(\nu+v) \xi^n + X^{(a)}\mu^\dag_a\xi^n\\
    \txi{a}& = \xi^{a} & 
    \underline{\tc}_n^\dag & = \underline{c}_n^\dag + [\epsilon_n, (Z_a \underline{\mu}^\dag)^{(a)}]\\
        && \underline{\txi{}}^\dag &= \underline{\xi}^{\dag} + \underline{c}_n^\dag (\nu+v) + d_{\omega}\epsilon_n (Z_a \underline{\mu}^\dag)^{(a)} + \nu [\epsilon_n, (Z_a \underline{\mu}^\dag)^{(a)}] \\
        && \underline{\txi{}}_n^{\dag} &= \underline{\xi}^{\dag}_n + X^{(a)}\underline{c}_{an}^\dag \mu^\dag + d(e^{N-4}\epsilon \underline{\tau}^\dag \mu^\dag) + \iota_z \underline{c}_n^\dag (\nu+v) + \iota_z(\nu+v) [\epsilon_n, (Z_a \underline{\mu}^\dag)^{(a)}]\\
        && &{\phantom{=}}+ d_{\omega_n}\epsilon_n (Z_a \underline{\mu}^\dag)^{(a)} + X^{(b)}\mu_{b}^\dag [\epsilon_n, (Z_a \underline{\mu}^\dag)^{(a)}]
    \end{align*}

where $\widetilde{Z}=Z= \omega^{\dag} + c_n^\dag \xi^n$ and $X= [c +\iota_{\xi}v, \epsilon_n]+ L_{\xi{}}^{\omega}(\epsilon_n)-d_{\omega_n} \epsilon_n\xi^{n}$ and $\theta_N(5)=0$ for $N=4$, $\theta_N(5)=\frac{1}{(N-4)!}$ for $N>4$.
The proof that $\psi^*(\varpi_H) = \varpi_s$ and $\psi^*(S_H)= S_s$ are long computations and are postponed to Appendix \ref{a:appendix}. Let us now just recall the key points (for $N=4$).

Using the decompositions \eqref{e:omegadecomp_d4},  $\omega= \widehat{\omega}+ v$ and \eqref{e:split_omegan_dag}, $\underline{\omega}_n^\dag= \underline{\widehat{\omega}}_n^\dag + \underline{v}^{\dag} $ in the standard BV gravity symplectic form \eqref{e:standard_symplectic_form_split}, we get
\begin{align} 
    \varpi_{s}
    &= \varpi_{r} + \int_{\Sigma \times I}\delta \widehat{\omega}  \delta \underline{v}^\dag +\delta v \delta \underline{\widehat{\omega}}^\dag_n +\delta v \delta  \underline{v}^\dag.
\end{align}

The difference with $\varpi_H$ \eqref{e:sympl_form_H} is then given by the terms $\delta \widehat{\omega}  \delta \underline{v}^\dag +\delta v \delta \underline{\widehat{\omega}}^\dag_n$. We can write the first as 
\begin{align*}
    \delta \widehat{\omega} \delta (\epsilon_n [e, \underline{\mu}^\dag]) = \delta ( \epsilon_n [\widehat{\omega},e]) \delta  \underline{\mu}^\dag + \cdots
\end{align*}
and then use \eqref{e:PCconstraints1}  getting 
\begin{align*}
    \delta ( \epsilon_n [\widehat{\omega},e]) \delta  \underline{\mu}^\dag = \delta (e \sigma) \delta  \underline{\mu}^\dag + \cdots
\end{align*} 
Similarly, for $\delta v \delta \underline{\widehat{\omega}}^\dag_n$ we can use \eqref{e:PCconstraints2} and get 
\begin{align*}
    \delta v \delta \underline{\widehat{\omega}}^\dag_n = \delta v \delta (e \underline{\tau}^\dag) + \cdots = \delta e \delta (v \tau^\dag) + \cdots .
\end{align*}
Then these two terms are compensated by the corresponding term arising from the transformation of $\underline{e}_n^\dag$.
The same constraints are used in the same way in the computation for the actions.
\end{proof}

\section{BV pushforward}\label{s:BV-pushforward_Gravity_Hedgehog}\label{s:hedge}

\subsection{The product structure}\label{s:gravity_hedgehog}
\begin{proposition}\label{p:gravity_hedgehog}
    Let $e_0$ be a reference tetrad and define $\mathcal{F}_f=T^*[1](V)$ where $V=\ker W_{e_0}$. Then $\mathcal{F}_H \rightarrow \mathcal{F}_r$ is a BV bundle as in Section \ref{s:BV_bundles} with typical fiber $\mathcal{F}_f$.
\end{proposition}

\begin{proof} 
    We define $\pi$ as the natural projection $\mathcal{F}_H \rightarrow \mathcal{F}_r$. We have to find a symplectomorphism $\pi^{-1}(\mathcal{F}_r)\rightarrow \mathcal{F}_r \times \mathcal{F}_f$. 
    For each $x \in M$, there exists a unique transformation $\Lambda_x$ such that $ \Lambda_x e = e_0$.\footnote{{Since for each point $\Lambda_x$ is unique and $e$ and $e_0$ are global, also $\Lambda$ is global.}} Correspondingly, we have $\ker W_{e_0}^{(2,1)}= \Lambda_x^2 \ker W_e^{(2,1)} $.
    We then apply $\Lambda^2_x$ to $v_x$ and extend this to $v^\dag$ as a cotangent lift.
\end{proof}

\subsection{Nondegeneracy of the Gaussian integrand}

We have now proven that the standard formulation of BV gravity is {strongly BV equivalent} to a BV product. We can then perform a Gaussian integration on the fiber in order to obtain the restricted theory. This can be done thanks to the following lemma:
\begin{lemma}\label{l:quadratic_nondegenerate}
    The quadratic form in $v$ 
\begin{align*}
    \int_{\Sigma \times I} \frac{1}{2(N-3)!}\underline{e_n}e^{N-3} [v,v]
\end{align*}
is nondegenerate. 
\end{lemma}

In order to prove this lemma, we fix a specific set of coordinates in which $e_{\mu}^{a}= \delta_{\mu}^{a}$ and $(e_n)^{a}= \delta_N^{a}$. Using this basis the components of $v$ have been characterized as follows:
\begin{lemma}\cite{CCS2020}\label{l:components_of_v}
    Let $v \in \Omega(\Sigma, \wedge^2 \mathcal{V})$ such that $ev=0$. Then in the basis for which $e_{\mu}^{a}= \delta_{\mu}^{a}$ and $(e_n)^{a}= \delta_N^{a}$, the components of $v$ satisfy the following set of equations:
    \begin{align*}
        v_{i}^{N j} =0 \qquad & 1 \leq i,j \leq N-1 \; i \neq j \\
        \sum_{i\neq N, i\neq j} v_{i}^{iN} =0 \qquad & 1 \leq j \leq N-1 \\
        \sum_{i\neq N, i\neq j} v_{i}^{ij} =0 \qquad & 1 \leq j \leq N-1.
    \end{align*}
\end{lemma}

\begin{proof}[Proof of Lemma \ref{l:quadratic_nondegenerate}.]
We start from $N=4$ (see also \cite{CC2025}).
    In components, we have the following expression of the quadratic form ($1\leq a,b,c,d,e,f \leq 4$, $1\leq \mu, \nu, \rho\leq 3$):
\begin{align*}
    (e_n)^{a} e_{\mu}^{b} [v,v]^{cd}_{\nu\rho}\epsilon^{\mu\nu\rho}\epsilon_{abcd} &=
    (e_n)^{a} e_{\mu}^{b} v_{\nu}^{ce}v_{\rho}^{df}\eta_{ef}\epsilon^{\mu\nu\rho}\epsilon_{abcd}\\
    &= \delta_4^{a}\delta_{\mu}^{b}v_{\nu}^{ce}v_{\rho}^{df}\eta_{ef}\epsilon^{\mu\nu\rho}\epsilon_{abcd}\\
    &= v_{\nu}^{ce}v_{\rho}^{df}\eta_{ef}\epsilon^{b\nu\rho}\epsilon_{bcd},
\end{align*}
where now we also have $1\leq b,c,d \leq 3$. Using now the identity
\begin{align*}
    \epsilon^{b\nu\rho}\epsilon_{bcd} = \delta_{c}^{\nu}\delta_{d}^{\rho} - \delta_{d}^{\nu}\delta_{c}^{\rho},
\end{align*}
we get
\begin{align*}
    (e_n)^{a} e_{\mu}^{b} [v,v]^{cd}_{\nu\rho}\epsilon^{\mu\nu\rho}\epsilon_{abcd} &= v_{\nu}^{ce}v_{\rho}^{df}\eta_{ef}\left(\delta_{c}^{\nu}\delta_{d}^{\rho} - \delta_{d}^{\nu}\delta_{c}^{\rho}\right)\\
    & = v_{c}^{ce}v_{d}^{df}\eta_{ef}- v_{d}^{ce}v_{c}^{df}\eta_{ef}.
\end{align*}
From Lemma \ref{l:components_of_v} we deduce that the free components of $v$ are parametrized by 
\begin{align}\label{e:v_components}
    &v_{1}^{23} &  &v_{2}^{13} & &v_{3}^{12} & &v_{1}^{12} & &v_{1}^{13} & &v_{2}^{21}
\end{align}
and we have $v_{3}^{32}= -v_{1}^{12}$, $v_{2}^{23}= -v_{1}^{13}$ and $v_{3}^{31}= -v_{2}^{21}$, while all the other components are zero.

Then we have (note that now also $1\leq e,f\leq 3$)
\begin{align*}
    (e_n)^{a} e_{\mu}^{b} [v,v]^{cd}_{\nu\rho}\epsilon^{\mu\nu\rho}\epsilon_{abcd} 
    & = v_{c}^{ce}v_{d}^{df}\eta_{ef}- v_{d}^{ce}v_{c}^{df}\eta_{ef}\\
    &= v_{c}^{c1}v_{d}^{d1} + v_{c}^{c2}v_{d}^{d2} + v_{c}^{c3}v_{d}^{d3}- v_{d}^{c1}v_{c}^{d1}-v_{d}^{c2}v_{c}^{d2}-v_{d}^{c3}v_{c}^{d3}\\
    &= v_{2}^{21}v_{2}^{21} + 2v_{2}^{21}v_{3}^{31}+ v_{3}^{31}v_{3}^{31} + v_{1}^{12}v_{1}^{12} + 2v_{1}^{12}v_{3}^{32} + v_{3}^{32}v_{3}^{32}\\
    &\phantom{=} + v_{1}^{13}v_{1}^{13}+ 2v_{2}^{23}v_{1}^{13}+ v_{2}^{23}v_{2}^{23}- v_{2}^{21}v_{2}^{21}- 2v_{2}^{31}v_{3}^{21}- v_{3}^{31}v_{3}^{31}\\
    &\phantom{=} -v_{1}^{12}v_{1}^{12}-2v_{1}^{32}v_{3}^{12}-v_{3}^{32}v_{3}^{32} -v_{1}^{13}v_{1}^{13}-2v_{1}^{23}v_{2}^{13} -v_{2}^{23}v_{2}^{23}\\
    &= 2 \left(v_{2}^{21}v_{3}^{31}+v_{1}^{12}v_{3}^{32}+ v_{2}^{23}v_{1}^{13}-v_{2}^{31}v_{3}^{21}-v_{1}^{32}v_{3}^{12}-v_{1}^{23}v_{2}^{13} \right)\\
    &= 2 \left(-(v_{2}^{21})^2-(v_{1}^{12})^2- (v_{1}^{13})^2-v_{2}^{13}v_{3}^{12}+v_{1}^{23}v_{3}^{12}-v_{1}^{23}v_{2}^{13} \right).
\end{align*}
Then, expressing the quadratic form as a matrix with respect to the basis \eqref{e:v_components}, we get
\begin{center}
\begin{align*}
    A=\left(
    \begin{tabular}{cccccc}
         0  & -1 & 1  & 0  & 0  & 0  \\
         -1 & 0  & -1 & 0  & 0  & 0  \\
         1  & -1 & 0  & 0  & 0  & 0  \\
         0  & 0  & 0  & -2 & 0  & 0  \\
         0  & 0  & 0  & 0  & -2 & 0  \\
         0  & 0  & 0  & 0  & 0  & -2  
    \end{tabular}
    \right).
\end{align*}
\end{center}
Then we can compute the determinant:
\begin{align*}
    \det A = - \det 
    \left(
    \begin{tabular}{ccc}
         0  & -1 & 1   \\
         -1 & 0  & -1  \\
         1  & -1 & 0   
    \end{tabular}
    \right)= - \det 
    \left(
    \begin{tabular}{cc}
         -1  & -1  \\
         1   & 0   
    \end{tabular}
    \right) - \det 
    \left(
    \begin{tabular}{cc}
         -1  &  0  \\
         1   & -1   
    \end{tabular}
    \right)= -2 \neq 0.
\end{align*}

Let now $N>4$. 
    In components, we have the following expression of the quadratic form ($1\leq a_1,\dots a_N, b,c \leq N$, $1\leq \mu_1, \dots \mu_{N-1} \leq N-1$):
\begin{align*}
    (e_n)^{a_1} e_{\mu_1}^{a_2} \dots e_{\mu_{N-3}}^{a_{N-2}} &[v,v]^{a_{N-1}a_{N}}_{\mu_{N-2}\mu_{N-1}}\epsilon^{\mu_1\dots\mu_{N-1}}\epsilon_{a_1\dots a_N} \\
    &=(e_n)^{a_1} e_{\mu_1}^{a_2} \dots e_{\mu_{N-3}}^{a_{N-2}}  v_{\mu_{N-2}}^{a_{N-1}b}v_{\mu_{N-1}}^{a_N c}\eta_{bc}\epsilon^{\mu_1\dots\mu_{N-1}}\epsilon_{a_1\dots a_N}\\
    &= \delta_N^{a_1}\delta_{\mu_1}^{a_2}\dots \delta_{\mu_{N-3}}^{a_{N-2}}  v_{\mu_{N-2}}^{a_{N-1}b}v_{\mu_{N-1}}^{a_N c}\eta_{bc}\epsilon^{\mu_1\dots\mu_{N-1}}\epsilon_{a_1\dots a_N}\\
    &= v_{\mu_{N-2}}^{a_{N-1}b}v_{\mu_{N-1}}^{a_N c}\eta_{bc}\epsilon^{a_2\dots a_{N-2}\mu_{N-2}\mu_{N-1}}\epsilon_{a_2\dots a_{N-2} a_{N-1} a_N}
\end{align*}
where now we also have $1\leq a_2\dots  a_N \leq N-1$. Using the identity
\begin{align*}
\epsilon^{a_2\dots a_{N-2}\mu_{N-2}\mu_{N-1}}\epsilon_{a_2\dots a_{N-2} a_{N-1} a_N}
     = \frac{1}{N-2}\left(\delta_{a_{N-1}}^{\mu_{N-2}}\delta_{a_N}^{\mu_{N-1}} - \delta_{a_N}^{\mu_{N-2}}\delta_{a_{N-1}}^{\mu_{N-1}}\right),
\end{align*}
and denoting by $\nu= \mu_{N-2}$, $\rho= \mu_{N-1}$, $c=a_{N-1}$, $d= a_{N}$, modulo a factor, we get
\begin{align*}
(e_n)^{a_1} e_{\mu_1}^{a_2} \dots e_{\mu_{N-3}}^{a_{N-2}} &[v,v]^{a_{N-1}a_{N}}_{\mu_{N-2}\mu_{N-1}}\epsilon^{\mu_1\dots\mu_{N-1}}\epsilon_{a_1\dots a_N} \\
     &= v_{\nu}^{ce}v_{\rho}^{df}\eta_{ef}\left(\delta_{c}^{\nu}\delta_{d}^{\rho} - \delta_{d}^{\nu}\delta_{c}^{\rho}\right)\\
    & = v_{c}^{ce}v_{d}^{df}\eta_{ef}- v_{d}^{ce}v_{c}^{df}\eta_{ef}\\
    & = v_{c}^{ce}v_{d}^{de}- v_{d}^{ce}v_{c}^{de}
\end{align*}
where in the last line $1\leq c,d,e  \leq N-1$, they are all different and summation over $c,d,e$ is understood.
From Lemma \ref{l:components_of_v} we deduce that the free components of $v$ are parametrized by 
\begin{align}\label{e:v_components2}
    &v_{c}^{de} &  &v_{c}^{ce} 
\end{align}
where $\sum_{c\neq N, c\neq e} v_{c}^{ce} =0$. Let us begin with the second term of the sum $-v_{d}^{ce}v_{c}^{de}$. If we select three different indexes $c,d,e$, the corresponding components $v_{c}^{de}$, $v_{d}^{ce}$, and $v_{e}^{cd}$ appear in the terms
\begin{align*}
    - v_{c}^{de} v_{d}^{ce}- v_{c}^{ed}v_{e}^{cd}- v_{d}^{ec}v_{e}^{dc} = - v_{c}^{de} v_{d}^{ce}+ v_{c}^{de}v_{e}^{cd}- v_{d}^{ce}v_{e}^{cd}.
\end{align*}
The matrix associated to this part of the quadratic form reads
\begin{align*}
    A=\frac{1}{2}\left(
    \begin{tabular}{ccc}
         0  & -1 & 1    \\
         -1 & 0  & -1   \\
         1  & -1 & 0    \\
    \end{tabular}
    \right),
\end{align*}
and its determinant is 1. Counting all possible triples $c,d,e$ we get $\binom{3}{N-1}$ blocks of this type.

Let us now consider the terms $v_{c}^{ce}v_{d}^{de}$ and let us fix $e$. Let now $m = \max (1, \dots, N-1)$ with $m \neq e$. Then from the equation $\sum_{c\neq N, c\neq e} v_{c}^{ce} =0$ we get 
\begin{align*}
    v_{m}^{me} = -\sum_{c\neq N, e, m} v_{c}^{ce},
\end{align*}
where all the other $v_{c}^{ce}$ are now independent. Hence we get
\begin{align*}
    \sum_{\substack{d,c\neq e \\ d \neq c}} v_{c}^{ce}v_{d}^{de} &= \sum_{\substack{d,c\neq e, m \\ d \neq c}} v_{c}^{ce}v_{d}^{de} + 2\sum_{c\neq e,m} v_{c}^{ce}v_{m}^{me} \\
    &= \sum_{\substack{d,c\neq e, m \\ d \neq c}} v_{c}^{ce}v_{d}^{de} -2 \sum_{c,d\neq e,m} v_{c}^{ce}v_{d}^{de}.
\end{align*}
Therefore the corresponding $(N-3) \times (N-3)$-matrix is 
\begin{align*}
    B=-\frac{1}{2}\left(
    \begin{tabular}{cccc}
         4  & 1  & $\hdots$ & 1    \\
         1  & 4  & $\hdots$ & 1    \\
         $\vdots$ & $\vdots$ & $\ddots$ & $\vdots$\\
         1  & 1 & $\hdots$ & 4    
    \end{tabular}
    \right).
\end{align*}
We can compute its determinant as the sum of an invertible matrix ($3$Id) and the outer product of a vector $w =(1, \cdots, 1)$ with its transposed:
\begin{align*}
    B = -\frac{1}{2}\left( 3 \mathrm{Id} + w w^T \right).
\end{align*}
Hence 
\begin{align*}
    \det B =\left( -\frac{1}{2}\right)^{N-3}\left(1 + \frac{1}{3}w^T \mathrm{Id} w\right)3^{N-3} \det \mathrm{Id} = \left(-\frac{3}{2}\right)^{N-3} \left(1 + \frac{N-3}{3}\right)= \left(-\frac{3}{2}\right)^{N-3} \frac{N}{3}.
\end{align*}
Note that in the matrix representing the quadratic form there are $(N-1)$ such blocks. Hence, since the determinant is always different from zero, we get the desired claim.
\end{proof}

Collecting all the information, we get the following theorem.

\begin{theorem}\label{t:BV_pushforward}
Let 
\begin{align*}
    \widetilde{\mathcal{P}_V} : \text{Dens}^{\frac{1}{2}} \mathcal{F}_s \rightarrow \text{Dens}^{\frac{1}{2}} \mathcal{F}_r
\end{align*}
be defined by
\begin{align*}
    \text{Dens}^{\frac{1}{2}} \mathcal{F}_s \overset{\psi^*}{\longrightarrow} \text{Dens}^{\frac{1}{2}} \mathcal{F}_H \overset{\mathcal{P}_V}{\longrightarrow} \text{Dens}^{\frac{1}{2}} \mathcal{F}_r
\end{align*}
where $\psi$ is defined respectively in Section \ref{s:proofTheorem} and $\mathcal{P}_V$ is the BV pushforward with respect to the Lagrangian submanifold $V   \coloneqq \{\underline{v}^\dag = 0\}$ inside {$\mathcal{F}_f$}. Then we have that
\begin{align*}
        \mu_r^{\frac{1}{2}} e^{\frac{i}{\hbar} S_r} = \widetilde{\mathcal{P}_V} \left( \mu_s^{\frac{1}{2}} e^{\frac{i}{\hbar} S_s}\right)
\end{align*}
where $\mu_r^{\frac{1}{2}}$ and $\mu_s^{\frac{1}{2}}$ are suitable half-densities on $\mathcal{F}_r$ and $\mathcal{F}_s$ respectively.
\end{theorem}

 This theorem shows that the restricted BV theory can be obtained as the BV pushforward of the standard BV theory. 

\begin{remark}
    In this paper we have asssumed the cosmological constant $\Lambda=0$. 
The same construction holds for a non vanishing $\Lambda$, just by adding a term 
\begin{align*}
    \int_{\Sigma \times I} \frac{1}{(N-1)!} \underline{e}_n e^{N-1} \Lambda
\end{align*}
to $S_r$.
\end{remark}

\newcommand{\calF}{{\mathsf F}}
\newcommand{\ccalF}{{\mathcal F}}
\newcommand{\ccalS}{{\mathcal S}}
{
\subsection{The quasiisomorphism}
The BV pushforward establishes a quasiisomorphism between two theories assuming that the fiber BV cohomology is trivial. This is the case for us, since the fiber differential is governed by the invertible quadratic form discussed in the previous section.

Let us summarize the general results (see \cite[Sec.\ 3.5]{Mne09} and \cite[Sec.\ 4.5.2]{Mne08}). We start with the case of a graded symplectic vector space $(\calF,\omega)$, with $\omega$ odd and of degree $-1$, which is the direct sum of the graded symplectic vector spaces $(\calF',\omega')$ and $(\calF'',\omega'')$. On $\calF$ we have a solution $S$ of the master equation of the form $S= S_0+I$, where $S_0$ is quadratic and $I$ is a perturbation. We assume that $S_0$ also satisfies the master equation (this is automatic if we assume that $I$ is at least cubic) and that it is a sum $S_0'+S_0''$ with $S_0'$ and $S_0''$ quadratic on $\calF'$ and $\calF''$, respectively (i.e., we admit no mixed terms). The cohomological vector fields $Q_0'=(S_0',\ )$ and $Q_0'=(S_0',\ )$ are then linear. To them we associate linear operators $\delta'\colon\calF'\to\calF'$ and $\delta''\colon\calF''\to\calF''$. The master equation implies that $\delta'$ and $\delta''$ are coboundary operators. The main assumption now is that $\delta''$ is acyclic, i.e., the cohomology of $(\calF'',\delta'')$ is trivial.\footnote{Later we will consider the cohomology of $Q_0''$ on an appropriate space of functions on $\calF''$. Acyclicity of $\delta''$ implies, by the Poincar\'e lemma, that the cohomology of $Q_0''$ consists only of the constant functions.} The usual spectral sequence for a double complex shows that the complexes $(\calF,\delta=\delta'+\delta'')$ and $(\calF',\delta')$ are quasiisomorphic.

More concretely, one chooses a special chain homotopy $K$, i.e., a linear map $\calF\to\calF$ of degree $-1$, satisfying
\[
[\delta,K] = \mathrm{Id} - p'\iota',\qquad K\iota' = 0, \qquad p'K=0, \qquad K^2=0,
\]
where $\iota'$ and $p'$ are the inclusion of $\calF'$ into $\calF$ and the projection from $\calF$ to $\calF'$.
The last conditions on $K$ together with the further condition that $K$ should be compatible with the symplectic structure imply that the image $\mathsf L$ of $K$ is a lagrangian subspace of $\calF''$, so we can also consider the associated BV pushforward. 
By homological perturbation theory, one shows that the BV pushforward $P$ (a deformation of the pullback map $(\iota')^*$ from functions on $\calF$ to functions on $\calF'$) is part of a special chain homotopy triple $(I,P,\mathsf{K})$ where $\mathsf{K}$ is the map on functions on $\calF$ determined by $K$ and $I$, a deformation of $(p')^*$,  may be reconstructed from the other data. The differentials are $(S,\ )-i\hbar\Delta$, on the functions on $\calF$, and $(S_\textrm{eff},\ )-i\hbar\Delta'$, on the functions on $\calF'$. Here $S_\textrm{eff}=S_0' + O(\hbar)$ is the effective action obtained by the BV pushforward.\footnote{The relation between the BV integral and homological perturbation lemma is clearly explained in \cite{GJF18}. Their result can be generalized to fiber BV integrals \cite{CM}. See also \cite{JKMSW20}.}

In the nonlinear case, we first pick a critical point of the BV action $\ccalS$ on $\ccalF$. We then denote by $\calF=\calF'\oplus\calF''$ the horizontal--plus--vertical tangent-space decomposition of $\ccalF=\ccalF_1\times\ccalF_2$ at the critical point and by $S$ the Taylor expansion of $\ccalS$ minus the value at the critical point. By proceeding as in the linear case, with the same assumptions, we get a quasiisomorphism between the complexes of jets of functions around the critical point. The effective action $S_\text{eff}$ will then be the Taylor expansion, minus the value at the critical point, of the action ${\mathcal S}_1$ determined by the BV pushforward.\footnote{We are not aware of a nonperturbative version of this result. On the other hand, we only consider the BV pushforward perturbatively.}

We now only have to check that the assumptions are satisfied in our case. We are only interested in critical points where the ghosts ans the antifields are zero. By the structure of our theory, we then see that such a critical point is a solution of the Euler--Lagrangian equations for Palatini--Cartan with $v=0$. It follows that $S_0''=\frac12\int_{\Sigma\times I}v''Av''$, where $v''$ is the perturbation of $v$ and $A$ is the quadratic form determined by $\be$ at the critical point.  Note that $Q_0''(v^\dagger)''=Av''$ and vanishes otherwise. Therefore, $\calF''$ at each point $x$ of $\Sigma\times I$ is the two-term complex $\calF''_{0}\oplus\calF''_{1}$, with $v''_x$ and  $(v^\dagger)''_x$ as coordinates and coboundary operator $\delta''$ given by $A_x$.
The results of the previous section show that $A_x$ is invertible (at every $x$ for every coframe $\be$). Therefore, the complex $(\calF'',\delta'')$ is acyclic. As special chain homotopy $K$, we can simply choose $A^{-1}$. 
This corresponds to a lagrangian submanifold with tangent space at the critical point given by the image of $K$ (i.e., $\calF''_{0}$), which is compatible with our choice ${\mathcal L}=\{v^\dagger=0\}$.

}

\section{Dealing with the boundary}\label{s:boundary}
In the previous sections, we have ignored the fact that the interval $I$ has boundary points. Two simple ways out are either to work on $\Sigma\times S^1$ or to take $I$ as an open interval (we can even take $I=\mathbb{R}$) and to assume fast decay of the fields towards the boundary of $I$. In any of these two cases, the results of the paper hold exactly as stated so far.

It is however interesting to consider also the case when $I$ is a compact interval (we also assume that $\Sigma$ is a closed manifold). In this case we have to deal with the boundary structure and to impose appropriate boundary conditions when performing the BV pushforward.

\subsection{The pre--BV-BFV formalism}
If a BV theory, for a closed space--time manifold, is put on a compact manifold with boundary, the master equation may be spoiled. Actually, what typically does not hold anymore is the condition that the cohomological vector field $Q$ is the hamiltonian vector field of the BV action $S$: in fact, the original equation $\iota_Q\omega=\delta S$, true for a closed space--time manifold, may now be spoiled by boundary terms arising by integration by parts in the computation of the functional derivative $\delta S$. We recollect here the details of the construction to circumvent this problem (this is an adapted reformulation of the construction of \cite{CMR2012,CMR2012b,CMR2015}).

\newcommand{\ualpha}{{\underline\alpha}}
\newcommand{\uvarpi}{{\underline\varpi}}
\newcommand{\uS}{{\underline S}}

\newcommand{\talpha}{{\widetilde\alpha}}
\newcommand{\tvarpi}{{\widetilde\varpi}}
\newcommand{\tS}{{\widetilde S}}
\newcommand{\tQ}{{\widetilde Q}}
\newcommand{\tF}{{\widetilde{\mathcal{F}}}}

\newcommand{\id}{{\mathrm{id}}}

\newcommand{\calY}{{\mathcal Y}}
\newcommand{\de}{\partial}

\begin{definition}[Relaxed BV manifold] 
A relaxed BV manifold is a quadruple $(\mathcal{F},S,\varpi,Q)$, where $\mathcal{F}$ is a graded manifold and $S,\varpi,Q$ are respectively a function, a closed $2$-form, and a vector field on $\mathcal{F}$ satisfying the following:
\begin{itemize}
\item $Q$ is cohomological, i.e., it is odd of degree $1$ and $[Q,Q]=0$;
\item $\varpi$ is odd of degree $-1$ and $S$ is even of degree $0$.\footnote{A generalization of this, called relaxed BF$^{k}$V manifold, replaces this condition by saying that $S$ has degree $k$, $\omega$ has degree $k-1$, and they both have parity equal to their degree modulo $2$.}
\end{itemize}
\end{definition}
\noindent Given this, we define
\begin{align*}
\ualpha & \coloneqq \iota_Q\varpi-\delta S,\\
\uvarpi & \coloneqq \delta\ualpha = - L_Q\varpi.
\end{align*}
Immediate consequences of Cartan's calculus, degree considerations, and the condition $[Q,Q]=0$ are the following:
\begin{enumerate}
\item $\ualpha$ and $\uvarpi$ are even of degree $0$.
\item There exists a unique odd function $\uS$ of degree $1$ satisfying $\iota_Q\uvarpi=\delta\uS$.
\end{enumerate}
This shows that $(\mathcal{F},\uS,\ualpha,Q)$ is an exact preBFV manifold according to the following
\begin{definition}[BFV]
A preBFV manifold is a quadruple $(\tF,\tS,\tvarpi,\tQ)$, where $\tF$ is a graded manifold and $\tS,\tvarpi,\tQ$ are respectively a function, a closed $2$-form, and a vector field on $\mathcal{F}$ satisfying the following:
\begin{itemize}
\item $\tQ$ is cohomological, i.e., it is odd of degree $1$ and $[\tQ,\tQ]=0$;
\item $\tvarpi$ is even of degree $0$ and $\tS$ is odd of degree $1$;
\item $\tQ$ is a hamiltonian vector field of $\tS$, i.e., $\iota_{\tQ}\tvarpi=\delta\tS$.
\end{itemize}
The preBFV manifold is called a BFV manifold if $\tvarpi$ is nondegenerate. If $\tvarpi$ is exact with a specified potential $\talpha$, i.e., $\tvarpi=\delta\talpha$, then we say that $(\tF,\tS,\talpha,\tQ)$ is an exact (pre)BFV manifold.
\end{definition}

We also have a relation between the original relaxed BV manifold $(\mathcal{F},S,\varpi,Q)$ and the resulting exact preBFV manifold $(\mathcal{F},\uS,\ualpha,Q)$ (namely, the defining equation $\ualpha = \iota_Q\varpi-\delta S$), which can be viewed as a special case of the following
\begin{definition}[BV-BFV]\label{d:preBV-BFV}
A preBV-BFV structure is a triple $((\mathcal{F},S,\varpi,Q),(\tF,\tS,\talpha,\tQ),\pi)$, where $(\mathcal{F},S,\varpi,Q)$ is a relaxed BV manifold, $(\tF,\tS,\talpha,\tQ)$ is an exact preBFV manifold, and $\pi\colon \mathcal{F}\to\tF$ is a surjective submersion, such that
\begin{itemize}
\item $Q$ and $\tQ$ are $\pi$-related, i.e., for every function $f$ on $\tF$, we have $Q\pi^*f = \pi^*\tQ f$;
\item $\iota_Q\varpi=\delta S+\pi^*\talpha$.
\end{itemize}
\end{definition}
It is then easy to compute how the ``master equation'' is violated in a general preBV-BFV structure:
\begin{theorem}\label{t:preBV-BFV}
Let $((\mathcal{F},S,\varpi,Q),(\tF,\tS,\talpha,\tQ),\pi)$ be a preBV-BFV structure. Then
\[
\frac12 \iota_Q\iota_Q\varpi = \pi^*\tS\qquad\text{and}\qquad
L_QS=\pi^*(2\tS-\iota_{\tQ}\talpha).
\]
\end{theorem}

We can summarize what we have so far by saying that $((\mathcal{F},S,\varpi,Q),(\mathcal{F},\uS,\ualpha,Q),\id)$ is a preBV-BFV structure. 

The reason for considering other (pre)BV-BFV structures is that the $2$-form $\uvarpi=\delta\ualpha$
is typically degenerate and we might want to remove part or all of the degeneracy by performing a quotient {as follows.}

{When $\uvarpi$ is degenerate, we can consider a simple, integrable distribution $\calY$} on $\mathcal{F}$ such that for every $Y$ in $\calY$ we have $\iota_Y\uvarpi=0$. Note that $Q$ is automatically projectable, i.e., $[Q,Y]\in\calY$ for every $Y$ in $\calY$. Let $\tF$ denote the leaf space of the distribution $\calY$ and $\pi\colon \mathcal{F}\to\tF$ the canonical projection. There are then uniquely determined $\tvarpi$ and $\tQ$ such that $\uvarpi=\pi^*\tvarpi$ and $Q\pi^*f = \pi^*\tQ f$ for every function $f$ on $\tF$. There is also a uniquely determined function $\tS$ satisfying $\iota_{\tQ}\tvarpi=\delta\tS$ and $\uS=\pi^*\tS$.
If there is a (necessarily unique) $1$-form $\talpha$ on $\tF$ satisfying $\ualpha=\pi^*\talpha$, then $((\mathcal{F},S,\varpi,Q),(\tF,\tS,\talpha,\tQ),\pi)$ is a preBV-BFV structure.\footnote{There is a more general version of a preBV-BFV structure, which we are not going to use in this paper, where $\tvarpi$ is only locally exact, i.e., it is the curvature of a connection $1$-form $\talpha$ for some line bundle over $\tF$. We can view $\ualpha$ as a connection $1$-form for the trivial line bundle over $\mathcal{F}$. Requiring that $\ualpha$ descend to the quotient as a connection is a weaker condition than requiring that it descend as a globally defined $1$-form.} 

In particular, let $\calY$ be the kernel of $\uvarpi$, i.e., the module of all vector fields $Y$ such that $\iota_Y\uvarpi=0$. This distribution is automatically involutive but not necessarily regular or simple. If it is, and $\ualpha$ is basic, we can perform the reduction as above. In this case, the resulting $2$-form $\tvarpi$ is nondegenerate, so symplectic. The notation ${(\mathfrak{F}, \mathfrak{F^\partial}, \pi^\partial)=}((\mathcal{F},S,\varpi,Q),(\mathcal{F}^\de,S^\de,\alpha^\de,Q^\de),\pi^\de)$ is reserved for the resulting BV-BFV structure.

{Hence, a theory on a manifold with boundary is in general not a BV theory but a preBV-BFV theory, hence we have to modify the BV pushforward in order to account for this relaxation. We will call this new version relaxed BV pushforward.}

\newcommand{\calB}{{\mathcal B}}

\subsection{Relaxed BV pushforward}
Let $(\mathcal{F},S,\varpi,Q)$ be a relaxed BV manifold. 
As a warm up, we are interested in the 
BV integral $Z_{\mathcal{L}} = \int_{\mathcal{L}} \left( \mu^{\frac{1}{2}} e^{\frac{i}{\hbar} S}\right)$, where $\mu$ is a compatible Berezinian for the graded symplectic manifold $(\mathcal{F},\varpi)$ and $\mathcal{L}$ is a Lagrangian submanifold. At this point, $Z_{\mathcal{L}}$ might not be invariant under deformations of $\mathcal{L}$ because $S$ possibly does not satisfy the quantum master equation. We assume for simplicity that $\Delta_\mu S=0$ (although $\Delta_\mu$ is regularized in the infinite-dimensional case). Still, the classical master equation is in general violated, in one of the forms appearing in Theorem~\ref{t:preBV-BFV}. To avoid the problem, we then have to make a choice so that $\uS$ and $\ualpha$ disappear (note that we consider the preBV-BFV structure in which $\tS=\uS$ and $\talpha=\ualpha$). This choice also has to be compatible with all the structure. We make the following
\begin{definition}
A good b-condition\footnote{b stands for boundary, for in the application to field theory this is really a boundary condition} is a submanifold $\calB\stackrel\iota\hookrightarrow\mathcal{F}$ such that
\begin{enumerate}
\item $\calB$ is a symplectic submanifold of $(\mathcal{F},\varpi)$;
\item $\calB$ is $Q$-invariant;
\item $\iota^*\uS=0$;
\item $\iota^*\ualpha=0$.
\end{enumerate}
\end{definition}
These conditions imply that $(\mathcal{B},\iota^*S,\iota^*\varpi,Q)$ is a BV manifold.\footnote{The first condition implies that $\iota^*\varpi$ is an odd symplectic form on $\calB$. The second condition implies that $Q$ is an odd vector field on $\calB$. The third condition implies that $Q$ on $\calB$ is Hamiltonian with Hamiltonian function $\iota^*S$. The last condition now implies that the master equation is satisfied, so $Q$ is a cohomological vector field on $\calB$.}
The last condition moreover implies that $\calB$ is an isotropic submanifold of $(\mathcal{F},\uvarpi)$. 
It is convenient to assume that it is maximally isotropic.\footnote{For applications in field theory, one usually wants to perturb around a critical point of $S$. Maximality of $\calB$ is usually required to ensure that the space of critical points is not empty.} Now it makes sense to consider $Z_{\mathcal{L}} = \int_{\mathcal{L}} \left( \mu^{\frac{1}{2}} e^{\frac{i}{\hbar} S}\right)$ for $\mu$  a compatible Berezinian and $\mathcal{L}$ the intersection with $\calB$ of the given Lagrangian submanifold of $\mathcal F$. The usual BV theorem now implies that $Z_{\mathcal{L}}$ is invariant under deformations of ${\mathcal{L}}$.

We now move to the BV pushforward.
Assume, as in Section~\ref{sGBVpf}, 
that $\mathcal{F}= \mathcal{F}_1 \times \mathcal{F}_2$ and $\varpi = \varpi_1 + \varpi_2$. We also assume, as in Corollary~\ref{c:BV-pushforward_QME}, that we have 
a compatible Berezinian for the graded symplectic manifold $(\mathcal{F},\varpi)$
and a lagrangian submanifold $\mathcal{L}$ of $\mathcal{F}_2$.
We are interested in the BV pushforward
\[
\mu_1^{\frac{1}{2}} e^{\frac{i}{\hbar} S_1} = \mathcal{P}_{\mathcal{L}} \left( \mu^{\frac{1}{2}} e^{\frac{i}{\hbar} S}\right).
\]
As before, the problem is that $S$ possibly does not satisfy the quantum master equation and we assume {as well} that $\Delta_\mu S=0$. Still, the classical master equation is in general violated. To avoid the problem, we then have to make a suitable choice, adapting the notion of good b-condition.

\begin{definition}\label{d:good_b_condition}
A good b-condition 
is a submanifold $\calB\stackrel\iota\hookrightarrow\mathcal{F}$ such that
\begin{itemize}
\item $\calB$ is a symplectic submanifold of $(\mathcal{F},\varpi)$;
\item $\calB$ is $Q$-invariant;
\item the restriction {
$p_\calB\colon\calB\to\mathcal{F}_1$ of $p\colon \mathcal{F}\to\mathcal{F}_1$ }is 
such that {$\iota_Y\uvarpi=0$} for all $p_\calB$-vertical vector fields $Y$;
\item {$\iota^*\uS=p_\calB^*\uS_1$} for a (necessarily uniquely determined) function $\uS_1$ on $\mathcal{F}_1$;
\item {$\iota^*\ualpha=p_\calB^*\ualpha_1$} for a (necessarily uniquely determined) $1$-form $\ualpha_1$ on $\mathcal{F}_1$.
\end{itemize}
\end{definition}

{\begin{remark}
    Note that the third condition guarantees that the good b-condition $\calB$ only imposes conditions on the fields that \emph{generate} the degeneracy of the two form $\uvarpi$. Hence, this means that a good b-condition is not a physical boundary condition on the phase space but rather imposes some condition so that a symplectic phase space can be constructed. 
\end{remark}}

We now work on $\calB$. Note the we have
\[
\Delta_\mu e^{\frac{i}{\hbar} S}{=} \frac12 \left(\frac{i}{\hbar}\right)^2 L_QS e^{\frac{i}{\hbar} S}
=
\left(2\pi^*\uS_1-\iota_{Q}\pi^*\ualpha_1\right)
\frac12 \left(\frac{i}{\hbar}\right)^2 e^{\frac{i}{\hbar} S}.
\]
The BV pushforward $\Omega_1$ of $2\pi^*\uS_1-\iota_{Q}\pi^*\ualpha_1$
along $\mathcal{L}$ determines the violation of the {quantum} master equation for $S_1$. 

In field theory, $\uS$ and $\ualpha$, and therefore $\uS_1$, $\ualpha_1$, and $\Omega_1$, only depend on boundary fields (and their transversal jets). As a consequence, the violation of the master equation for $S_1$ also has this property. Moreover, a deformation of ${\mathcal{L}}$ leads to a change of $e^{\frac{i}{\hbar} S_1}$ by a $\Delta_{\mu_1}$-exact term and by a term that only depends on boundary fields (and their transversal jets).

For simplicity, we focus on the case,  occurring in the current paper, that the $\mathcal{F}_1$-component  $Q_1$ of $Q$
depends only on ${\mathcal F}_2$. In this case, $\Omega_1=2\uS_1-\iota_{Q_1}\ualpha_1$. We also assume that
$S_1$ is the evaluation of $S$ at the critical point. As a result, $((\mathcal{F}_1,S_1,\varpi_1,Q_1),(\mathcal{F}_1,\uS_1,\ualpha_1,Q_1),\id)$ is the preBV-BFV structure for the effective theory.
Moreover, a deformation of $\mathcal L$ changes
$\mu_1^{\frac{1}{2}} e^{\frac{i}{\hbar} S_1}$ by an $\left(\Omega_1-i\hbar\Delta_{\mu_1}\right)$-exact term.

The BV pushforward now has all the properties as in Section~\ref{sGBVpf} (c.f. in particular Theorem~\ref{t:BV-pushforward_properties}).

\subsection{Back to PC gravity}\label{s:backtoPC}
In presence of a boundary, we get a preBV-BFV structure to which we can apply a relaxed BV pushforward. 
\begin{proposition}\label{p:gravity_preBVBFV}
    Denoting by $\mathfrak{F}^\partial=(\mathcal{F}^\partial, \varpi^\partial, S^\partial)$ the BFV theory described in \cite[Theorem 30]{CCS2020}, the PC gravity preBV-BFV structure is given by the following data (listed as in Definition \ref{d:preBV-BFV}):
\begin{itemize}
    \item The bulk space of fields is $\mathcal{F}=  \mathcal{F}_r \times \mathcal{F}_2 .$
    \item The bulk symplectic form is   
        \begin{align*}
        \varpi = \varpi_{r} + \int_{\Sigma \times I}\delta v \delta  \underline{v}^\dag.
        \end{align*}
    \item The bulk action is 
        \begin{align*}
        S  = S_{r} + \int_{\Sigma \times I} \frac{1}{2(N-3)!}\underline{e_n}e^{N-3} [v,v] 
        +\left( L_{\xi}^{\omega} v + [c,v] + \xi^{n} d_{\omega_n} v  + d\xi^n \iota_z v \right)\underline{v}^\dag. 
    \end{align*}
    \item $\tF = \mathcal{F}^\partial \times \mathcal{F}_{add}$ 
    where \begin{align*}
        \mathcal{F}_{add} = \Omega^{1}(\Sigma, \wedge^2\mathcal{V}) \oplus \Omega^{N-2}(\Sigma, \wedge^{N-2}\mathcal{V}). 
        \end{align*}
    \item $\tvarpi = \delta \widetilde{\alpha} = \varpi^{\partial} + \varpi_{add}$ where
        \begin{align*}
        \varpi_{add} = \int_{\Sigma} \delta v \delta (v^\dag \xi^n).
        \end{align*}
    \item $\tS = S^{\partial} + S_{add}$
        where
        \begin{align*}
        S^{\partial}_{add} &= \int_{\Sigma}\frac{1}{2(N-3)!}e_n\xi^n e^{N-3} [v,v ]+ \left( L_{\xi}^{\omega} v + [c,v] + d\xi^n \iota_z v \right)v^\dag \xi^n.
        \end{align*}
    \item The projection $\pi$ is given by {$(\pi^\partial, \iota^\partial): \mathcal{F}_H \times \mathcal{F}_2 \rightarrow \mathcal{F}^\partial \times \mathcal{F}_{add}$} where {$\pi^\partial$} is the composition of the inverse of the symplectomorphism described in \cite[equation (78)]{CCS2020b} and the projection to the boundary {and $\iota_\partial$ is the restriction of $\mathcal{F}_2$ to the boundary.}
\end{itemize}
\end{proposition}
\begin{proof}
    The bulk data is that of Theorem \ref{t:split_sympl_form}.
    Using the results in \cite{CCS2020b} it is a long but easy computation to show that these data form a preBV-BFV structure.
\end{proof}

In order now to perform a relaxed BV pushforward, we then have to take a good b-condition.
We define
\begin{align}\label{e:gravity_b_condition}
    {\mathcal B} \coloneqq \{ \iota_\partial^*v = 0 \text{ and } \iota_\partial^*(d_{\omega_n})^kv=0\,\forall k \in \mathbb{N}^*
\} {\subseteq \mathcal{F}_2 \subseteq \mathcal{F}_H\times \mathcal{F}_2}.
\end{align}

\begin{proposition}
    $\mathcal B$ is a good b-condition.
\end{proposition}

\begin{proof}
{We first note that $p_\mathcal{B}$ is the identity, hence the third condition is automatically satisfied. Then we have $\iota^* \widetilde{S}= S^\partial$ (which corresponds to $\underline{S}_1$ in the definition \ref{d:good_b_condition}) and $\iota^* \widetilde{\alpha}= \alpha^\partial$ (which corresponds to $\underline{\alpha}_1$ in the definition \ref{d:good_b_condition}). We are left with the first condition (trivially verified) and the second which}
    requires some additional care. From the expression of $S_H$, we deduce that 
    \begin{align*}
        Qv &=  L_{\xi}^{\omega} v + [c,v] + \xi^{n} d_{\omega_n} v  + d\xi^n \iota_z v\\
        Q \omega_n &= \iota_{\xi} F_{\omega_n} + d_{\omega_n} c 
    \end{align*}
    Hence $Q v |_{\mathcal{B}}=0$ and 
    \begin{align*}
        Q \left(d_{\omega_n}v\right) = [\iota_{\xi} F_{\omega_n} + d_{\omega_n} c ,v] + d_{\omega_n} \left(L_{\xi}^{\omega} v + [c,v] + \xi^{n} d_{\omega_n} v  + d\xi^n \iota_z v\right).
    \end{align*}
    Hence, using $\iota_\partial^*(d_{\omega_n})^kv=0$, it is easy to see, that $\mathcal B$ is $Q$-invariant.
\end{proof}

With this choice of a good b-condition, we can proceed and perform a relaxed BV pushforward exactly as in the case without boundary, with the same results.

{
In order to characterize this result we need a generalization of strong BV-equivalence as follows:
\begin{definition}\label{def:2-strongBVeq}
A strong BV-BFV equivalence between two BV-BFV theories $(\mathfrak{F}_a,\mathfrak{F}_a^\partial, \pi_a^\partial)$ and $(\mathfrak{F}_b,\mathfrak{F}_b^\partial, \pi_b^\partial)$ is a pair of degree-$0$ symplectomorphisms
\begin{align*}
    \Phi &: (\mathcal{F}_a, \varpi_a) \rightarrow (\mathcal{F}_b, \varpi_b)\\
    \Phi^{\partial} &: (\mathcal{F}_a^{\partial}, \varpi_a^{\partial} ) \rightarrow (\mathcal{F}_b^{\partial}, \varpi_b^{\partial})
\end{align*}
preserving the BV and BFV action: $\Phi^* S_b = S_a$ and $\Phi^{\partial*} S_b^{\partial} = S_a^{\partial}$ and satisfying
$$ \pi_b^{\partial} \circ \Phi = \Phi^{\partial} \circ \pi_a^{\partial} .$$
\end{definition}

Summarizing, let $I$ have boundary points. Then, denoting by $\mathfrak{F}_{\mathrm{NEW}}$ the relaxed BV-pushforward of $(\mathcal{F}_H, \varpi_H, S_H)$ as in Theorem \ref{t:split_sympl_form}, with respect to the Lagrangian submanifold $V \subset \mathcal{F}_f$ (Theorem \ref{t:BV_pushforward} and Proposition \ref{p:gravity_hedgehog}) with good b-condition \eqref{e:gravity_b_condition} we have the following theorem:

\begin{theorem}\label{t:BV-pushforward_boundary}
    The BV-BFV theory $(\mathfrak{F}_{\mathrm{NEW}}, \mathfrak{F}^\partial, \pi^\partial))$
   is strongly BV-BFV equivalent to the one given by $(\mathfrak{F}_{\AKSZ}, \mathfrak{F}^\partial, \pi^\partial)$.
\end{theorem}

Hence we showed how to recover a BV-BFV theory of gravity using the relaxed BV-pushforward.
}

\appendix 
\section{Detailed proof of Theorem \ref{t:split_sympl_form}} \label{a:appendix}

\subsection{Symplectic forms}
We prove here that 
\begin{align}
    \psi^* (\varpi_H)= \varpi_s.
\end{align}
We have:
\begin{align*}
    \varpi_H = \int_{\Sigma \times I} \delta \te \delta \underline{\te}_n^{\dag} + \delta \underline{\te}_n \delta \te^\dag  + \delta \widetilde{\whom} \delta \underline{\widetilde{\whom}}_n^{\dag} + \delta \underline{\tom}_n \delta \tom^\dag + \delta \tc \delta \underline{\tc}_n^{\dag} + \iota_{\delta \txi{}} \delta \underline{\txi{}}^\dag + \delta \txi{n} \delta \underline{\txi{}}_n^\dag + \delta \tv \delta \underline{\tv}^{\dag}
\end{align*}
\begin{align}\label{e:omega_s}
        \varpi_s &= \int_{\Sigma \times I} 
        \unl{\delta e \delta \underline{e}_n^{\dag} }{psiS1}
        + \unl{\delta \underline{e}_n \delta e^\dag  }{psiS2}
        + \unl{\delta \whom \delta \underline{\whom}_n^{\dag} }{psiS3} 
        + \unl{\delta \underline{\omega}_n \delta \omega^\dag }{psiS4}
        + \unl{\delta c \delta \underline{c}_n^{\dag} }{psiS5}
        + \unl{\iota_{\delta \xi} \delta \underline{\xi}^\dag }{psiS6}
        + \unl{\delta \xi^{n} \delta \underline{\xi}_n^\dag }{psiS7}\\
    &\phantom{= \int_{\Sigma \times I} } \nonumber
        + \unl{\delta v \delta \underline{v}^{\dag} }{psiS8}
        + \unl{\delta \whom \delta \underline{v}^{\dag} }{psiS9}
        + \unl{\delta v\delta \underline{\whom}_n^{\dag}}{psiS10}
\end{align}

Hence, using the explicit expression of $\psi$ we get

\setcounter{terms}{0} 
\begin{align}\label{e:psi-omegaH}
    \psi^* (\varpi_H) 
    &= \int_{\Sigma \times I} 
    \unl{\delta e \delta \underline{e}_n^\dag}{psiHe1} 
    + \frac{1}{N-3}\unl{\delta e \delta  \left(  e^{N-4}d_{\widehat{\omega}} (\epsilon_n \underline{\mu}^\dag) +  e^{N-4}\sigma \underline{\mu}^\dag\right)}{psiHe2} 
    + \frac{1}{N-3}\unl{\delta e \delta \left(e^{N-4})\underline{\tau}^\dag \nu\right)}{psiHe3}\\
    &\phantom{=  } \nonumber
    + \frac{1}{N-3}\unl{\delta e \delta \delta \left(e^{N-4})\underline{\tau}^\dag v\right)}{psiHe4}
    + \unl{\delta e \delta \left(X^{(a)}(Z_b \underline{\mu}^\dag)_a^{(b)}\right)}{psiHe5}
    + \unl{\delta e \delta \left(\iota_{\underline{z}}(v_a Z)^{(a)}\right)}{psiHe6}\\
    &\phantom{=  } \nonumber
    + \unl{\delta e \delta \left(\iota_{\underline{z}}(\nu_a Z)^{(a)}\right)}{psiHe7}
    + \theta_N(5) \unl{\delta e \delta \left(e^{N-5}\epsilon_n \mu^\dag d\xi^n \underline{\tau}^\dag\right)}{psiHe8}\\
    &\phantom{=  } \nonumber
    +\unl{\delta \underline{e}_n \delta e^\dag }{psiHen1}
    +\unl{\delta \underline{e}_n \delta (v_a Z)^{(a)} }{psiHen2}
    +\unl{\delta \underline{e}_n \delta (\nu_a Z)^{(a)} }{psiHen3}\\
    &\phantom{=  } \nonumber
    +\unl{\delta \whom \delta \underline{\whom}_n^\dag}{psiHo1}
    +\unl{\delta \nu \delta \underline{\whom}_n^\dag}{psiHo2}
    +\unl{\delta \whom \delta [\epsilon_n, \iota_{\xi}(Z_a \underline{\mu}^\dag)^{(a)}] }{psiHo3}
    +\unl{\delta \nu \delta [\epsilon_n, \iota_{\xi}(Z_a \underline{\mu}^\dag)^{(a)}] }{psiHo4}\\
    &\phantom{=  } \nonumber
    +\unl{\delta \underline{\omega}_n \delta \omega^\dag }{psiHon1}
    +\unl{\delta \iota_{\underline{z}}\nu \delta \omega^\dag }{psiHon2}
    +\unl{\delta \iota_{\underline{z}}v \delta \omega^\dag }{psiHon3}
    +\unl{\delta X^{(a)}\underline{\mu}^\dag_a \delta \omega^\dag }{psiHon4}\\
    &\phantom{=  } \nonumber
    +\unl{\delta \underline{\omega}_n \delta [\epsilon_n, (Z_a \mu^\dag)^{(a)}\xi^n] }{psiHon5}
    +\unl{\delta \iota_{\underline{z}}\nu \delta [\epsilon_n, (Z_a \mu^\dag)^{(a)}\xi^n] }{psiHon6}
    +\unl{\delta \iota_{\underline{z}}v \delta [\epsilon_n, (Z_a \mu^\dag)^{(a)}\xi^n] }{psiHon7}\\
    &\phantom{=  } \nonumber
    +\unl{\delta X^{(a)}\underline{\mu}^\dag_a [\epsilon_n, (Z_a \mu^\dag)^{(a)}\xi^n] }{psiHon8}
    +\unl{\delta c \delta \underline{c}_n^\dag}{psiHc1}
    +\unl{\delta (\iota_{\xi}\nu) \delta \underline{c}_n^\dag}{psiHc2}
    +\unl{\delta (\iota_{\xi}v) \delta \underline{c}_n^\dag}{psiHc3}
    +\unl{\delta (\iota_{z}\nu \xi^n) \delta \underline{c}_n^\dag}{psiHc4}\\
    &\phantom{=  } \nonumber
    +\unl{\delta (\iota_{z} v \xi^n) \delta \underline{c}_n^\dag}{psiHc5}
    +\unl{\delta \left(X^{(a)}\mu^\dag_a\xi^n\right) \delta \underline{c}_n^\dag}{psiHc6}
    +\unl{\delta c \delta [\epsilon_n, (Z_a \underline{\mu}^\dag)^{(a)}]}{psiHc7}
    +\unl{\delta (\iota_{\xi}\nu) \delta [\epsilon_n, (Z_a \underline{\mu}^\dag)^{(a)}]}{psiHc8}\\
    &\phantom{=  } \nonumber
    +\unl{\delta (\iota_{\xi}v) \delta [\epsilon_n, (Z_a \underline{\mu}^\dag)^{(a)}]}{psiHc9}
    +\unl{\delta (\iota_{z}\nu \xi^n) \delta [\epsilon_n, (Z_a \underline{\mu}^\dag)^{(a)}]}{psiHc10}
    +\unl{\delta (\iota_{z} v \xi^n) \delta [\epsilon_n, (Z_a \underline{\mu}^\dag)^{(a)}]}{psiHc11}\\
    &\phantom{=  } \nonumber
    +\unl{\delta \left(X^{(a)}\mu^\dag_a\xi^n\right) \delta [\epsilon_n, (Z_a \underline{\mu}^\dag)^{(a)}]}{psiHc12}
    + \unl{\iota_{\delta \xi} \delta \underline{\xi}^{\dag}}{psiHx1}
    + \unl{\iota_{\delta \xi} \delta (\underline{c}_n^\dag \nu)}{psiHx2}
    + \unl{\iota_{\delta \xi} \delta (\underline{c}_n^\dag v)}{psiHx3}\\
    &\phantom{=  } \nonumber
    + \unl{\iota_{\delta \xi} \delta (d_{\omega}\epsilon_n (Z_a \underline{\mu}^\dag)^{(a)})}{psiHx4}
    + \unl{\iota_{\delta \xi} \delta (\nu [\epsilon_n, (Z_a \underline{\mu}^\dag)^{(a)}])}{psiHx5}\\
    &\phantom{=  } \nonumber
    + \unl{\delta \xi^n \delta \underline{\xi}^{\dag}_n}{psiHxn1}
    + \unl{\delta \xi^n \delta (X^{(a)}\underline{c}_{an}^\dag \mu^\dag)}{psiHxn2}
    + \unl{\delta \xi^n \delta (d(e^{N-4}\epsilon \underline{\tau}^\dag \mu^\dag))}{psiHxn3}
    + \unl{\delta \xi^n \delta (\iota_z \underline{c}_n^\dag \nu)}{psiHxn4}\\
    &\phantom{=  } \nonumber
    + \unl{\delta \xi^n \delta (\iota_z \underline{c}_n^\dag v)}{psiHxn5}
    + \unl{\delta \xi^n \delta (\iota_z\nu [\epsilon_n, (Z_a \underline{\mu}^\dag)^{(a)}])}{psiHxn6}
    + \unl{\delta \xi^n \delta (d_{\omega_n}\epsilon_n (Z_a \underline{\mu}^\dag)^{(a)})}{psiHxn7}\\
    &\phantom{=  } \nonumber
    + \unl{\delta \xi^n \delta (X^{(b)}\mu_{b}^\dag [\epsilon_n, (Z_a \underline{\mu}^\dag)^{(a)}])}{psiHxn8}
    + \unl{\delta \xi^n \delta (\iota_zv [\epsilon_n, (Z_a \underline{\mu}^\dag)^{(a)}])}{psiHxn9}
    + \unl{\delta v \delta \underline{v}^\dag}{psiHv1}.
\end{align}

Let us now compare the two expressions \eqref{e:omega_s} and \eqref{e:psi-omegaH}  and show that each term appears in both expressions. Let us divide the computation into some steps.
\begin{enumerate}[label=\Alph*]
\item The following terms appear verbatim in the two expressions: 
\begin{align*}
    \reft{e:omega_s}{psiS1} &= \reft{e:psi-omegaH}{psiHe1};& 
        \reft{e:omega_s}{psiS2} &= \reft{e:psi-omegaH}{psiHen1};& 
        \reft{e:omega_s}{psiS3} &= \reft{e:psi-omegaH}{psiHo1};\\
        \reft{e:omega_s}{psiS4} &= \reft{e:psi-omegaH}{psiHon1};& 
        \reft{e:omega_s}{psiS5} &= \reft{e:psi-omegaH}{psiHc1};& 
        \reft{e:omega_s}{psiS6} &= \reft{e:psi-omegaH}{psiHx1};\\
        \reft{e:omega_s}{psiS7} &= \reft{e:psi-omegaH}{psiHxn1};& 
        \reft{e:omega_s}{psiS8} &= \reft{e:psi-omegaH}{psiHv1}.
\end{align*}
        
\item The following terms in \eqref{e:psi-omegaH} are equal and simplify:
    \begin{align*}
        \reft{e:psi-omegaH}{psiHo4} &= \reft{e:psi-omegaH}{psiHc8} + \reft{e:psi-omegaH}{psiHx5}; &
        \reft{e:psi-omegaH}{psiHon6} &= \reft{e:psi-omegaH}{psiHc10} + \reft{e:psi-omegaH}{psiHxn6};\\
        \reft{e:psi-omegaH}{psiHon7} &= \reft{e:psi-omegaH}{psiHc11} + \reft{e:psi-omegaH}{psiHxn9}; &
        \reft{e:psi-omegaH}{psiHon8} &= \reft{e:psi-omegaH}{psiHc12} + \reft{e:psi-omegaH}{psiHxn8}.
    \end{align*}

\item \label{i:etau}Consider \reft{e:psi-omegaH}{psiHe4}:
    \begin{align*}
        \frac{1}{N-3}\delta e \delta \left(e^{N-4} v {\tau}^\dag\right)&=\delta e^{N-3} \delta ( v \tau^\dag)= \delta \left(e^{N-3} \tau^\dag\right) \delta v\\
        & = \delta \underline{\whom}^\dag_n \delta v - \delta( \iota_{\underline{z}}Z )\delta v +\delta (\iota_{\xi}\underline{c}_n^{\dag} ) \delta v,
    \end{align*}
    The first and last terms are respectively \reft{e:omega_s}{psiS10} and \reft{e:psi-omegaH}{psiHc3} $+$ \reft{e:psi-omegaH}{psiHx3}, while the second is compensated by \reft{e:psi-omegaH}{psiHe6}, \reft{e:psi-omegaH}{psiHen2}, \reft{e:psi-omegaH}{psiHon3}, \reft{e:psi-omegaH}{psiHc5} and \reft{e:psi-omegaH}{psiHxn5}:
    \begin{align*}
    \delta (\iota_{\underline{z}}v) \delta \omega^\dag - \delta ( \iota_z v \xi^{n}) \delta \underline{c}_n^\dag + \delta \xi^{n} \delta (\iota_z v \underline{c}_{n}^\dag)= \delta (\iota_{\underline{z}}v) \delta(Z)
    \end{align*}
    and 
    \begin{align*}
        \delta e \delta \left(\iota_{\underline{z}}(v_a Z)^{(a)}\right) + \delta \underline{e}_n \delta \left((v_a Z)^{(a)}\right)= \iota_{\delta \underline{z}} \delta ( vZ)
    \end{align*}
\item Consider \reft{e:psi-omegaH}{psiHe2}:
\begin{align*}
      \frac{1}{N-3} &\delta e \delta  \left(  e^{N-4}d_{\widehat{\omega}} (\epsilon_n \underline{\mu}^\dag) +  e^{N-4}\sigma \underline{\mu}^\dag\right) =\delta e^{N-3} \delta \left(  d_{\widehat{\omega}} (\epsilon_n \underline{\mu}^\dag) +  \sigma \underline{\mu}^\dag\right)\\
     &=   \delta e^{N-3}  [\delta {\whom} , \epsilon_n \underline{\mu}^\dag] + d_{\whom} (\delta e^{N-3})  \epsilon_n \delta \underline{\mu^\dag} + \delta (e^{N-3} \sigma) \delta\underline{\mu}^\dag\\
     &\overset{\eqref{e:PCconstraints2}}{=}\delta {\widehat{\omega}} \delta \left(\epsilon_n[e^{N-3}, \underline{\mu}^\dag]\right)+\epsilon_n  \delta \underline{\mu}^\dag [\delta {\widehat{\omega}} ,  e^{N-3}]+ d_{\widehat{\omega}} \delta e^{N-3}  \epsilon_n \delta \underline{\mu}^\dag\\
     &\phantom{=}+ \delta \left({\epsilon}_n  d_{\whom} e^{(N-3)} - {\epsilon}_n e^{(N-4)} \tau^\dag d \xi^{n} + {X}^{[a]}Z_a\right) \delta\underline{\mu}^\dag\\
     &= \delta {\widehat{\omega}} \delta v^\dag+ \delta \left( - {\epsilon}_n e^{(N-4)} \tau^\dag d \xi^{n} + X^{[a]}Z_a\right) \delta\underline{\mu}^\dag.
\end{align*}
The first term is \reft{e:omega_s}{psiS9}. We consider the last terms in the following points.
\item We consider $\delta \left(X^{[a]}Z_a\right) \delta\underline{\mu}^\dag  $ together with \reft{e:psi-omegaH}{psiHon4}, \reft{e:psi-omegaH}{psiHc6}, \reft{e:psi-omegaH}{psiHxn2}:
\begin{align*}
    \delta  X^{[a]} \delta \left(Z_a \underline{\mu}^\dag \right) =   X^{[b]} (\delta e_b)^{[a]}\delta \left(Z_a \underline{\mu}^\dag \right) +  (\delta X)^{[a]} \delta \left(Z_a \underline{\mu}^\dag \right).
\end{align*}
The first term is  \reft{e:psi-omegaH}{psiHe5}, while, recalling that 
\begin{align*}
    \delta X = [\delta (c+ \iota_{\xi}v), \epsilon_n]+ [\iota_{\xi} \delta \whom, \epsilon_n] + \iota_{\delta \xi{}}d_{\whom}(\epsilon_n)-[\delta \omega_n, \epsilon_n]\xi^{n}+d_{\omega_n} \epsilon_n\delta \xi^{n}
\end{align*}
the second is given by \reft{e:psi-omegaH}{psiHc7}, \reft{e:psi-omegaH}{psiHo3}, \reft{e:psi-omegaH}{psiHx4}, \reft{e:psi-omegaH}{psiHon5}, \reft{e:psi-omegaH}{psiHc9} and \reft{e:psi-omegaH}{psiHxn7}.
\item Lastly we consider $\delta \left( - {\epsilon}_n e^{(N-4)} \tau^\dag d \xi^{n}\right)\delta\underline{\mu}^\dag$ together with \reft{e:psi-omegaH}{psiHxn3} and \reft{e:psi-omegaH}{psiHe8}:
\begin{align*}
   \delta &\left( - {\epsilon}_n e^{(N-4)} \tau^\dag d \xi^{n}\right)\delta\underline{\mu}^\dag+ \delta \xi^n \delta (d(e^{N-4}\epsilon \underline{\tau}^\dag \mu^\dag))+ \theta_N(5) \delta e \delta \left(e^{N-5}\epsilon_n \mu^\dag d\xi^n \underline{\tau}^\dag\right) \\
   &= \delta \tau^\dag \delta \left({\epsilon}_n e^{(N-4)}d \xi^{n}\underline{\mu}^\dag\right)= \delta \tau^\dag \delta \left( e^{(N-3)}\nu\right)\\
   & = \delta \left(\tau^\dag e^{(N-3)} \right) \delta \nu + \frac{1}{N-3}\delta e \delta\left(e^{N-4} \nu \tau^\dag\right).
\end{align*}
 The last term is \reft{e:psi-omegaH}{psiHe3}, while for the first we proceed exactly as in \ref{i:etau} and get \reft{e:psi-omegaH}{psiHo2}, \reft{e:psi-omegaH}{psiHc2}, \reft{e:psi-omegaH}{psiHx2}, \reft{e:psi-omegaH}{psiHe7}, \reft{e:psi-omegaH}{psiHen3}, \reft{e:psi-omegaH}{psiHon2}, \reft{e:psi-omegaH}{psiHc4} and \reft{e:psi-omegaH}{psiHxn4}.
\end{enumerate}

\subsection{Actions}

We use the same strategy for the action. We rewrite $S_s$ as the sum of $S_r$ and an additional part:
\setcounter{terms}{0} 
\begin{align}\label{e:standard_action_split3}
    S_{s}(\Sigma \times I) &= S_{r}(\Sigma \times I) 
    + \int_{\Sigma \times I} \unl{\frac{1}{(N-3)!}e^{N-3} \underline{e}_n d_{\whom} v}{term1}
    +\unl{ \frac{1}{2(n-3)!}e^{N-3} \underline{e}_n[v,v]}{term2} \\
    \nonumber&\phantom{= + \int}
    +\unl{ \frac{1}{(N-2)!} e^{N-2}  d_{\underline{\omega_n}}v}{term3}-\unl{  \iota_{\xi} d_{\whom} v \underline{\whom}_n^\dag}{term4}
    +\unl{ \frac{1}{2}\iota_{\xi}[v,v])\underline{\whom_n^\dag}}{term5}
    -\unl{ [v,c] \underline{\whom}_n^\dag}{term6}
    +\unl{ X(\underline{v^\dag}) }{term7}\\
    \nonumber&\phantom{= + \int}
    -\unl{ \iota_{\xi}  d_{\underline{\omega_n}}v\omega^{\dag}}{term8}
    +\unl{ [\iota_{\xi}v, e] \underline{e}_n^{\dag}}{term9}+\unl{ [\iota_{\xi}v, \underline{e}_n]e^\dag}{term10}-\unl{  \frac{1}{2}\iota_{\xi}\iota_{\xi}d_{\whom} v \underline{c}_n^\dag}{term11}+\unl{ \frac{1}{4}\iota_{\xi}\iota_{\xi}[v,v]\underline{c}_n^\dag}{term12}\\
    \nonumber&\phantom{= + \int}-\unl{d_{\omega_n}v \xi^n \underline{\whom}_n^\dag}{term13}-\unl{\iota_{\xi}d_{\omega_n}v\xi^n\underline{c}_n^\dag}{term14}
\end{align}
where \reft{e:standard_action_split3}{term7} reads
\begin{align*}
    X(\underline{v^\dag}) = \left( \unl{\iota_{\xi} F_{\whom} }{term15}+ \unl{\iota_{\xi} d_{\whom} v }{term16}+ \unl{[\iota_{\xi}v,v]}{term17}+ \unl{d_{\whom} \omega_n \xi^n}{term18} + \unl{\partial_n \whom \xi^n }{term19}+ \unl{d_{\omega_n}v \xi^n }{term20} - \unl{d_{\whom } c }{term21}- \unl{[v,c] }{term22}\right)\underline{v^\dag}
\end{align*}

Before computing the pullback of the action $S_H$ note the following:
\begin{align*}
\psi^*\left(\iota_{\txi{}}F_{\whom} + F_{\tom_n}\txi{n} + d_{\whom}\tc\right) &= \iota_{\xi}F_{\whom} + F_{\omega_n}\xi^n + d_{\whom}c + L_{\xi}^{\omega} \nu + \iota_z (\nu+v) d\xi^n\\
&\phantom{=}+ X^{(a)}\mu_a^\dag d\xi^n  + d_{\omega_n}\nu \xi^n + [\nu,c]\\
\psi^*\left(\iota_{\txi{}}\underline{F_{\tom_n}} +\underline{d_{\tom_n}}\tc \right) &= \iota_{\xi}\underline{F_{\omega_n}} +\underline{d_{\omega_n}}c + \iota_{\underline{\partial_n} \xi} \nu + \underline{d_{\omega_n}}\iota_\xi v + \underline{d_{\omega_n}}\left(X^{(a)}\mu_a^\dag \xi^n\right)\\
&\phantom{=} + \underline{d_{\omega_n}}\left(\iota_{\underline{z}} \nu + \iota_{\underline{z}} v \right) + [X^{(a)}\underline{\mu}_a^\dag,c] + \iota_{\xi}d_{\whom} \left(\iota_{\underline{z}} \nu + \iota_{\underline{z}} v \right)\\
&\phantom{=} + \iota_{\xi}d_{\whom} \left(X^{(a)}\underline{\mu}_a^\dag\right)+ [\iota_{\underline{z}} \nu + \iota_{\underline{z}} v,c]\\
\psi^*\left(\frac{1}{2}\iota_{\txi{}}\iota_{\txi{}}F_{\whom} + \iota_{\txi{}}F_{\tom_n}\txi{n} + \frac{1}{2}[\tc,\tc]\right) &= \frac{1}{2}\iota_{\xi} \iota_{\xi}F_{\whom} +  \iota_{\xi}F_{\omega_n}\xi^n + \frac{1}{2}[c,c]+ \frac{1}{2}\iota_{\xi} \iota_{\xi}d_{\whom} \nu \\
&\phantom{=}+ \iota_{\xi}d_{\whom} \left(\iota_z v+ \iota_z \nu\right)\xi^n + \iota_{\xi}d_{\whom} \left(X^{(a)}\mu_a^\dag\right)\xi^n + \iota_\xi d_{\omega_n}\nu \xi^n\\
&\phantom{=} + [X^{(a)}\mu_a^\dag,c] +[\iota_\xi (\nu+v),c]+ [\iota_z (\nu+v)\xi^n,c]
\end{align*}

We divide the computation in several parts. Let us start from the classical action:
\setcounter{terms}{0} 
\begin{align}\label{e:psiSH_cl}
\psi^* &\left(\frac{1}{(N-2)!}\te^{N-2} F_{\tom_n}+ \frac{1}{(N-3)!} \te^{N-3} \underline{\te}_n F_{\widetilde{\whom}}\right) = \frac{1}{(N-2)!} e^{N-2} F_{\omega_n}+ \frac{1}{(N-3)!} e^{N-3} \underline{e}_n F_{\whom} \\ \nonumber
    & + 
    \frac{1}{(N-2)!}e^{N-2}\left( \unl{ d_{\whom} \left(X^{(a)}\underline{\mu}_a^\dag\right)}{psiSHcl1}
    +  \unl{ d_{\whom} \left(\iota_{\underline{z}}(\nu+v)\right)}{psiSHcl2}
    +\unl{ \underline{d_{\omega_n}} \nu}{psiSHcl3}
    +  \unl{ [\nu, X^{(a)}\underline{\mu}_a^\dag]}{psiSHcl4}+  \unl{[\nu,\iota_{\underline{z}}(\nu+v)]}{psiSHcl5}\right)\\ \nonumber
    &
    + \frac{1}{(N-3)!} \unl{e^{N-3} \underline{e}_n d_{\whom} \nu}{psiSHcl6}
    + \frac{1}{2(N-3)!} \unl{e^{N-3} \underline{e}_n [\nu, \nu]}{psiSHcl7}
\end{align}
We then consider the terms with antifields $\whom_n^\dag$ and $\omega^\dag$
\setcounter{terms}{0} 
\begin{align}\label{e:psiSH_o}
    \psi^*&\left((\iota_{\txi{}}F_{\whom} + F_{\tom_n}\txi{n} + d_{\whom}\tc)\widetilde{\whom}_n^\dag\right) = \left(\iota_{\xi}F_{\whom} + F_{\omega_n}\xi^n + d_{\whom}c \right) \whom_n^\dag \\ \nonumber
    &+ \left(\unl{L_{\xi}^{\omega} \nu }{psiSHon1}
    + \unl{\iota_z (\nu+v) d\xi^n }{psiSHon2}
    + \unl{X^{(a)}\mu_a^\dag d\xi^n }{psiSHon3}
    + \unl{d_{\omega_n}\nu \xi^n }{psiSHon4} 
    + \unl{[\nu,c] }{psiSHon5}
    + \unl{[\nu,\iota_{\xi}v] }{psiSHon14}
    + \unl{d_{\whom} \iota_{\xi}v}{psiSHon15}\right) \underline{\whom}_n^\dag\\ \nonumber
    & + \left(\unl{\iota_{\xi}F_{\whom}}{psiSHon6} 
    + \unl{F_{\omega_n}\xi^n }{psiSHon7}
    + \unl{d_{\whom}c }{psiSHon8} 
    + \unl{ L_{\xi}^{\omega} \nu }{psiSHon9}
    + \unl{\iota_z (\nu+v) d\xi^n }{psiSHon10}
    + \unl{X^{(a)}\mu_a^\dag d\xi^n  }{psiSHon11}\right) [\epsilon_n, \iota_{\xi}(Z_a \underline{\mu}^\dag)^{(a)}]\\ \nonumber
    & + \left(
     \unl{d_{\omega_n}\nu \xi^n  }{psiSHon12}
    + \unl{[\nu,c]}{psiSHon13}
    + \unl{[\nu,\iota_{\xi}v]}{psiSHon16}
    + \unl{d_{\whom} \iota_{\xi}v}{psiSHon17}
    \right)[\epsilon_n, \iota_{\xi}(Z_a \underline{\mu}^\dag)^{(a)}]\\\nonumber
    \psi^*&\left((\iota_{\txi{}}\underline{F_{\tom_n}} +\underline{d_{\tom_n}}\tc )\tom^\dag\right) = \left(\iota_{\xi}\underline{F_{\omega_n}} +\underline{d_{\omega_n}}c \right)\omega^\dag 
    +\left( \unl{\iota_{\underline{\partial_n} \xi} \nu}{psiSHo1} 
    + \unl{\underline{d_{\omega_n}}\iota_\xi v}{psiSHo2} \right)\omega^\dag\\ \nonumber
    &+ \left( \unl{\underline{d_{\omega_n}}\left(X^{(a)}\mu_a^\dag \xi^n\right)}{psiSHo3}
    + \unl{ \underline{d_{\omega_n}}\left(\iota_{\underline{z}} \nu \xi^n+ \iota_{\underline{z}} v \xi^n\right)}{psiSHo4}
    + \unl{[X^{(a)}\underline{\mu}_a^\dag,c + \iota_{\xi}v]}{psiSHo5} \right)\omega^\dag \\ \nonumber
    & + \left(
    \unl{\iota_{\xi}d_{\whom} \left(\iota_{\underline{z}} \nu + \iota_{\underline{z}} v \right)}{psiSHo6}
    +\unl{\iota_{\xi}d_{\whom} \left(X^{(a)}\underline{\mu}_a^\dag\right)}{psiSHo7}
    + \unl{[\iota_{\underline{z}} \nu + \iota_{\underline{z}} v,c + \iota_{\xi}v]}{psiSHo8}
    \right)\omega^\dag\\ \nonumber
    & + \left(
    \unl{\iota_{\xi}\underline{F_{\omega_n}}}{psiSHo10} +\unl{\underline{d_{\omega_n}}c }{psiSHo11}
    +\unl{\iota_{\underline{\partial_n} \xi} \nu}{psiSHo12} 
    + \unl{\underline{d_{\omega_n}}\iota_\xi v}{psiSHo13}
    + \unl{\underline{d_{\omega_n}}\left(X^{(a)}\mu_a^\dag \xi^n\right)}{psiSHo14}
    \right) [\epsilon_n, (Z_a \mu^\dag)^{(a)}\xi^n] \\ \nonumber
    &+ \left( 
    \unl{ \underline{d_{\omega_n}}\left(\iota_{\underline{z}} \nu \xi^n+ \iota_{\underline{z}} v \xi^n\right)}{psiSHo15}
    + \unl{[X^{(a)}\underline{\mu}_a^\dag,c + \iota_{\xi}v]}{psiSHo16} 
    + \unl{\iota_{\xi}d_{\whom} \left(\iota_{\underline{z}} \nu + \iota_{\underline{z}} v \right)}{psiSHo17}
    \right) [\epsilon_n, (Z_a \mu^\dag)^{(a)}\xi^n]  \\ \nonumber
    & + \left(
    \unl{\iota_{\xi}d_{\whom} \left(X^{(a)}\underline{\mu}_a^\dag\right)}{psiSHo18}
    + \unl{[\iota_{\underline{z}} \nu + \iota_{\underline{z}} v,c + \iota_{\xi}v]}{psiSHo19}
    \right) [\epsilon_n, (Z_a \mu^\dag)^{(a)}\xi^n] 
\end{align}
We now consider the terms with antifields $e_n^\dag$ and $e^\dag$:
\setcounter{terms}{0} 
\begin{align}\label{e:psiSH_e}
    \psi^*&\left(\left(L_{\txi{}}^{\widetilde{\whom} } \te  + d_{\tom_n} \te \txi{n}  + \te_n d \txi{n} - [\tc,\te] \right)\underline{\te}_n^{\dag}\right) = \left(L_{\xi{}}^{\omega }e  + d_{\omega_n} e \xi^n  + e_n d \xi^n - [c,e] \right)\underline{e}_n^{\dag}\\ \nonumber
    &+ \unl{[\iota_{\xi}v,e]\underline{e}_n^{\dag}}{psiSHe21} + \frac{1}{N-3} e^{N-4} L_{\xi{}}^{\omega }e \left(  \unl{d_{\widehat{\omega}} (\epsilon_n \underline{\mu}^\dag) + \sigma \underline{\mu}^\dag}{psiSHe1}
    + \unl{\underline{\tau}^\dag (\nu+v)}{psiSHe2}\right)\\ \nonumber
    &+L_{\xi{}}^{\omega }e \left(\unl{X^{(a)}(Z_b \underline{\mu}^\dag)_a^{(b)}}{psiSHe3} 
    + \unl{\iota_{\underline{z}}((\nu_a+v_a) Z)^{(a)}}{psiSHe4}
    + \unl{\frac{\theta_N(5) e^{N-5}}{N-4}\epsilon_n \mu^\dag d\xi^n \underline{\tau}^\dag}{psiSHe5}\right)
    \\ \nonumber
    & + \frac{1}{N-3} e^{N-4} d_{\omega_n} e \xi^n \left(  \unl{d_{\widehat{\omega}} (\epsilon_n \underline{\mu}^\dag) + \sigma \underline{\mu}^\dag}{psiSHe6}
    + \unl{\underline{\tau}^\dag (\nu+v)}{psiSHe7}\right)\\ \nonumber
    &+d_{\omega_n} e \xi^n \left(\unl{X^{(a)}(Z_b \underline{\mu}^\dag)_a^{(b)}}{psiSHe8} 
    + \unl{\iota_{\underline{z}}((\nu_a+v_a) Z)^{(a)}}{psiSHe9}
    + \unl{\frac{\theta_N(5) e^{N-5}}{N-4}\epsilon_n \mu^\dag d\xi^n \underline{\tau}^\dag}{psiSHe10}\right)
    \\ \nonumber
    & + \frac{1}{N-3} e^{N-4} e_n d \xi^n \left(  \unl{d_{\widehat{\omega}} (\epsilon_n \underline{\mu}^\dag) + \sigma \underline{\mu}^\dag}{psiSHe11}
    + \unl{\underline{\tau}^\dag (\nu+v)}{psiSHe12}\right)\\ \nonumber
    &+e_n d \xi^n \left(\unl{X^{(a)}(Z_b \underline{\mu}^\dag)_a^{(b)}}{psiSHe13} 
    + \unl{\iota_{\underline{z}}((\nu_a+v_a) Z)^{(a)}}{psiSHe14}
    + \unl{\frac{\theta_N(5) e^{N-5}}{N-4}\epsilon_n \mu^\dag d\xi^n \underline{\tau}^\dag}{psiSHe15}\right)
    \\ \nonumber
    & + \frac{1}{N-3} e^{N-4} [c + \iota_{\xi} v,e] \left(  \unl{d_{\widehat{\omega}} (\epsilon_n \underline{\mu}^\dag) + \sigma \underline{\mu}^\dag}{psiSHe16}
    + \unl{\underline{\tau}^\dag (\nu+v)}{psiSHe17}\right)\\ \nonumber
    &+[c + \iota_{\xi} v,e] \left(\unl{X^{(a)}(Z_b \underline{\mu}^\dag)_a^{(b)}}{psiSHe18} 
    + \unl{\iota_{\underline{z}}((\nu_a+v_a) Z)^{(a)}}{psiSHe19}
    + \unl{\frac{\theta_N(5) e^{N-5}}{N-4}\epsilon_n \mu^\dag d\xi^n \underline{\tau}^\dag}{psiSHe20}\right)\\ \nonumber
    \psi^*&\left(\left(L_{\txi{}}^{\widetilde{\whom} } \underline{\te}_n  + \underline{d_{\tom_n}} (\te_n \txi{n})  + \iota_{\underline{\partial_n}\txi{}}\te - [\tc,\te_n] \right)\te^{\dag}\right) = \left(\iota_{\xi{}}d_{\omega}\underline{e}_n  + \iota_{\underline{\partial}_n \xi{}}e  - \underline{d_{\omega_n}}(e_n \xi^n) - [c,\underline{e}_n] \right)e^\dag \\ \nonumber
    &+ \unl{[\iota_{\xi}v,\underline{e}_n]e^{\dag}}{psiSHen5}+\left(\unl{\iota_{\xi{}}d_{\omega}\underline{e}_n }{psiSHen1} + \unl{\iota_{\underline{\partial}_n \xi{}}e }{psiSHen2} - \unl{\underline{d_{\omega_n}}(e_n \xi^n)}{psiSHen3} - \unl{[c+\iota_\xi v,\underline{e}_n]}{psiSHen4} \right)((\nu_a + v_a) Z)^{(a)}
\end{align}

Next we consider the terms with $c_n^\dag$
\setcounter{terms}{0} 
\begin{align}\label{e:psiSH_c}
    \psi^*&\left(\left(\frac{1}{2}\iota_{\txi{}}\iota_{\txi{}}F_{\whom} + \iota_{\txi{}}F_{\tom_n}\txi{n} + \frac{1}{2}[\tc,\tc]\right) \underline{\tc}_n^\dag\right) = \left(\frac{1}{2}\iota_{\xi} \iota_{\xi}F_{\whom} +  \iota_{\xi}F_{\omega_n}\xi^n + \frac{1}{2}[c,c]\right) \underline{c}_n^\dag \\ \nonumber
    &+\left( \frac{1}{2}\unl{\iota_{\xi} \iota_{\xi}d_{\whom} \nu }{psiSHc1}
    + \unl{\iota_{\xi}d_{\whom} \left(\iota_z v+ \iota_z \nu\right)\xi^n }{psiSHc2}
    + \unl{\iota_{\xi}d_{\whom} \left(X^{(a)}\mu_a^\dag\right)\xi^n }{psiSHc3}
    + \unl{\iota_\xi d_{\omega_n}\nu \xi^n }{psiSHc4}\right)\underline{c}_n^\dag \\ \nonumber
    &+ \left(\unl{[X^{(a)}\mu_a^\dag\xi^n,c]}{psiSHc5}
    +\unl{[\iota_\xi (\nu+v),c]}{psiSHc6}
    + \unl{[\iota_z (\nu+v)\xi^n,c]}{psiSHc7}\right)\underline{c}_n^\dag \\ \nonumber
     &+ \left(\unl{\frac{1}{4}\iota_{\xi}\iota_{\xi}[v,v]}{psiSHc18}+\unl{\frac{1}{2}\iota_{\xi}\iota_{\xi}[\nu,v]}{psiSHc19}+\unl{[\iota_{z}(\nu+v)\xi^n,\iota_{\xi}v]}{psiSHc20}+\unl{[X^{(a)}\mu_a^\dag \xi^n,\iota_{\xi}v]}{psiSHc21}\right)\underline{c}_n^\dag \\ \nonumber
    &+\left(\frac{1}{2}\unl{\iota_{\xi} \iota_{\xi}F_{\whom} }{psiSHc8}
    +\unl{  \iota_{\xi}F_{\omega_n}\xi^n }{psiSHc9}
    + \unl{ \frac{1}{2}[c,c]}{psiSHc10}\right) [\epsilon_n, (Z_a \underline{\mu}^\dag)^{(a)}] \\ \nonumber
    &+\left( \frac{1}{2}\unl{\iota_{\xi} \iota_{\xi}d_{\whom} \nu }{psiSHc11}
    + \unl{\iota_{\xi}d_{\whom} \left(\iota_z v+ \iota_z \nu\right)\xi^n }{psiSHc12}
    + \unl{\iota_{\xi}d_{\whom} \left(X^{(a)}\mu_a^\dag\right)\xi^n }{psiSHc13}
    + \unl{\iota_\xi d_{\omega_n}\nu \xi^n }{psiSHc14}\right)[\epsilon_n, (Z_a \underline{\mu}^\dag)^{(a)}] \\ \nonumber
    &+ \left(\unl{[X^{(a)}\mu_a^\dag\xi^n,c]}{psiSHc15}
    +\unl{[\iota_\xi (\nu+v),c]}{psiSHc16}+ \unl{[\iota_z (\nu+v)\xi^n,c]}{psiSHc17}\right)[\epsilon_n, (Z_a \underline{\mu}^\dag)^{(a)}]\\ \nonumber
    &+ \left(\unl{\frac{1}{4}\iota_{\xi}\iota_{\xi}[v,v]}{psiSHc22}+\unl{\frac{1}{2}\iota_{\xi}\iota_{\xi}[\nu,v]}{psiSHc23}+\unl{[\iota_{z}(\nu+v)\xi^n,\iota_{\xi}v]}{psiSHc24}+\unl{[X^{(a)}\mu_a^\dag \xi^n,\iota_{\xi}v]}{psiSHc25}\right)[\epsilon_n, (Z_a \underline{\mu}^\dag)^{(a)}] 
\end{align}
We now consider the terms with $\xi_n^\dag$ and $\xi^\dag$:
\setcounter{terms}{0} 
\begin{align}\label{e:psiSH_x}
    \psi^*&\left(\iota_{\txi{}}d\txi{n}\underline{\txi{\dag}}_n + \txi{n}\partial_n \txi{n}\underline{\txi{\dag}}_n\right)= \iota_{\xi}d\xi^n \underline{\xi}_n^\dag + \xi^n\partial_n \xi^n \underline{\xi}_n^\dag\\ \nonumber
    &+ \iota_{\xi}d\xi^n \left( \unl{X^{(a)}\underline{c}_{an}^\dag \mu^\dag }{psiSHxn1}
    + \unl{ d(e^{N-4}\epsilon \underline{\tau}^\dag \mu^\dag) }{psiSHxn2}
    + \unl{ \iota_z \underline{c}_n^\dag (\nu+v) }{psiSHxn3}
    + \unl{ \iota_z(\nu+v) [\epsilon_n, (Z_a \underline{\mu}^\dag)^{(a)}]}{psiSHxn4}\right)\\ \nonumber
    &+ \iota_{\xi}d\xi^n \left( \unl{ d_{\omega_n}\epsilon_n (Z_a \underline{\mu}^\dag)^{(a)}}{psiSHxn5}
    + \unl{ X^{(b)}\mu_{b}^\dag [\epsilon_n, (Z_a \underline{\mu}^\dag)^{(a)}]}{psiSHxn6}\right)\\ \nonumber
    &+ \xi^n\partial_n \xi^n \left( \unl{X^{(a)}\underline{c}_{an}^\dag \mu^\dag }{psiSHxn7}
    + \unl{ d(e^{N-4}\epsilon \underline{\tau}^\dag \mu^\dag) }{psiSHxn8}
    + \unl{ \iota_z \underline{c}_n^\dag (\nu+v) }{psiSHxn9}
    + \unl{ \iota_z(\nu+v) [\epsilon_n, (Z_a \underline{\mu}^\dag)^{(a)}]}{psiSHxn10}\right)\\ \nonumber
    &+ \xi^n\partial_n \xi^n \left( \unl{ d_{\omega_n}\epsilon_n (Z_a \underline{\mu}^\dag)^{(a)}}{psiSHxn11}
    + \unl{ X^{(b)}\mu_{b}^\dag [\epsilon_n, (Z_a \underline{\mu}^\dag)^{(a)}]}{psiSHxn12}\right)
    \\ \nonumber
    \psi^*&\left(\txi{n}\iota_{\partial_n\txi{}}\underline{\txi{\dag}} + \frac{1}{2}\iota_{[\txi{},\txi{}]}\underline{\txi{\dag}}\right) = \xi^{n}\iota_{\partial_n\xi}\underline{\xi}^{\dag} + \frac{1}{2}\iota_{[\xi,\xi]}\underline{\xi}^{\dag}\\ \nonumber
    & + \xi^{n}\iota_{\partial_n\xi}\left(\unl{\underline{c}_n^\dag (\nu+v) }{psiSHxn13}
    + \unl{ d_{\omega}\epsilon_n (Z_a \underline{\mu}^\dag)^{(a)} }{psiSHxn14}
    + \unl{\nu [\epsilon_n, (Z_a \underline{\mu}^\dag)^{(a)}]}{psiSHxn15}\right)
    \\ \nonumber
    & + \frac{1}{2}\iota_{[\xi,\xi]}\left(\unl{\underline{c}_n^\dag (\nu+v) }{psiSHxn16}
    + \unl{ d_{\omega}\epsilon_n (Z_a \underline{\mu}^\dag)^{(a)} }{psiSHxn17}
    + \unl{\nu [\epsilon_n, (Z_a \underline{\mu}^\dag)^{(a)}]}{psiSHxn18}\right)
\end{align}
Lastly we consider all the remaining terms:
\setcounter{terms}{0} 
\begin{align}\label{e:psiSH_v}
        \psi^*&\left(\frac{1}{2(N-3)!}\te^{N-3}\underline{\te}_n [\tv,\tv]+\left( L_{\txi{}}^{\widetilde{\whom}} \tv + [\tc,\tv] + \txi{n} d_{\tom_n} v  + d\txi{n} \iota_z \tv \right)\underline{\tv}^\dag\right)\\ \nonumber
        &= \unl{\frac{1}{2(N-3)!}e^{N-3}\underline{e}_n[v,v]}{psiSHv11}\left( \unl{L_{\xi}^{\omega} v }{psiSHv1}
        + \unl{[c,v] }{psiSHv2}
        + \unl{\xi^{n} d_{\omega_n} v  }{psiSHv3}
        + \unl{d\xi^n \iota_z v }{psiSHv4}
        + \unl{[\iota_{\xi}v,v] }{psiSHv9}\right)\underline{v}^\dag
    \end{align}
The terms that are not underlined correspond to the ones in $S_r$. Everything else should be found in \eqref{e:standard_action_split3} or vanish.

\begin{enumerate}[label=\Alph*]
    \item The following terms simplify:
    \begin{align*}
    \reft{e:standard_action_split3}{term2} &= \reft{e:psiSH_v}{psiSHv11};
    &\reft{e:standard_action_split3}{term3} &= 0;
    &\reft{e:psiSH_cl}{psiSHcl5} &=0;
    &\reft{e:psiSH_cl}{psiSHcl4} &=0;\\
    \reft{e:psiSH_o}{psiSHo10} &= \reft{e:psiSH_c}{psiSHc9};
    &\reft{e:psiSH_o}{psiSHon10} &= \reft{e:psiSH_x}{psiSHxn4};
    &\reft{e:psiSH_o}{psiSHon11} &= \reft{e:psiSH_x}{psiSHxn6};
    &\reft{e:psiSH_o}{psiSHon12} &= \reft{e:psiSH_c}{psiSHc14};\\
    \reft{e:standard_action_split3}{term22} &= \reft{e:psiSH_v}{psiSHv2};
    &\reft{e:standard_action_split3}{term20} &= \reft{e:psiSH_v}{psiSHv3};
    &\reft{e:standard_action_split3}{term17} &= \reft{e:psiSH_v}{psiSHv9};
    &\reft{e:standard_action_split3}{term9} &= \reft{e:psiSH_e}{psiSHe21};\\
    \reft{e:standard_action_split3}{term10} &= \reft{e:psiSH_e}{psiSHen5};
    &\reft{e:psiSH_o}{psiSHon16} &= \reft{e:psiSH_c}{psiSHc23};
    &\reft{e:psiSH_o}{psiSHo14} &= \reft{e:psiSH_x}{psiSHxn12};
    &\reft{e:psiSH_o}{psiSHo15} &= \reft{e:psiSH_x}{psiSHxn10};\\
    \reft{e:psiSH_o}{psiSHo12} &=  \reft{e:psiSH_x}{psiSHxn15};
    &\reft{e:psiSH_o}{psiSHo17} &= \reft{e:psiSH_c}{psiSHc12};
    &\reft{e:psiSH_o}{psiSHo18} &= \reft{e:psiSH_c}{psiSHc13};
        \end{align*}
    \begin{align*}
        \reft{e:psiSH_o}{psiSHo16} &= \reft{e:psiSH_c}{psiSHc15} + \reft{e:psiSH_c}{psiSHc25};
    &\reft{e:psiSH_o}{psiSHo19} &= \reft{e:psiSH_c}{psiSHc17} + \reft{e:psiSH_c}{psiSHc24};\\
     \reft{e:psiSH_c}{psiSHc11} &+ \reft{e:psiSH_o}{psiSHon9}  = \reft{e:psiSH_x}{psiSHxn18};
    &\reft{e:standard_action_split3}{term8} &+ \reft{e:psiSH_o}{psiSHo1} + \reft{e:psiSH_o}{psiSHo2} + \reft{e:psiSH_x}{psiSHxn13} = \reft{e:psiSH_e}{psiSHen2};\\
    \reft{e:psiSH_cl}{psiSHcl5} &+ \reft{e:psiSH_cl}{psiSHcl7} = \reft{e:psiSH_v}{psiSHv4}.
    \end{align*}

    \item \label{i:terms_with_c}
    We consider the terms containing the ghost $c$. Using \eqref{e:PCconstraints1} and \eqref{e:PCconstraints2} we get:
    \begin{align}
   &\reft{e:psiSH_e}{psiSHe16}+\reft{e:psiSH_e}{psiSHe17} + \reft{e:psiSH_e}{psiSHe20}
    \\
    &= d_{\whom} (c + \iota_{\xi} v) \underline{v}^\dag + [c + \iota_{\xi} v, X^{(a)}Z_a]\underline{\mu}^\dag + \frac{1}{N-3}[c + \iota_{\xi} v, \epsilon_n]^{(a)}\underline{\mu}_a^\dag e^{N-3} \left(d_{\whom}e + \tau^\dag d \xi^n\right) \nonumber\\ \label{e:item_terms_with_c}
    &\phantom{=}+ [c+ \iota_{\xi} v, \epsilon_n]^{(n)}X^{(a)}Z_a\underline{\mu}^\dag+  \frac{1}{N-3}[c + \iota_{\xi} v, v+\nu] e^{N-3} \underline{\tau}^\dag.
    \end{align}
    We also have that
    \begin{align*}
        \reft{e:psiSH_o}{psiSHo5} + \reft{e:psiSH_c}{psiSHc5}+\reft{e:psiSH_c}{psiSHc21} = [c, X^{(a)}\underline{\mu}^\dag_a] Z + [\iota_{\xi}v, X^{(a)}\underline{\mu}^\dag_a] Z.
    \end{align*}
    compensating the second term in \eqref{e:item_terms_with_c} and that
    \begin{align*}
        \reft{e:psiSH_o}{psiSHon5}&+ \reft{e:psiSH_o}{psiSHon14} + \reft{e:standard_action_split3}{term5}+\reft{e:standard_action_split3}{term6} + \reft{e:psiSH_o}{psiSHo8} +\reft{e:psiSH_c}{psiSHc6}+\reft{e:psiSH_c}{psiSHc7}+\reft{e:psiSH_c}{psiSHc18}+\reft{e:psiSH_c}{psiSHc19}+\reft{e:psiSH_c}{psiSHc20}\\
        &=  \frac{1}{N-3}[c+ \iota_{\xi}v, v+\nu] e^{N-3} \underline{\tau^\dag} 
    \end{align*}
    Compensating the last term of \eqref{e:item_terms_with_c}. Furthermore the following identity holds:
    \begin{align*}
        \reft{e:psiSH_c}{psiSHc10}&+ \reft{e:psiSH_o}{psiSHo13}+  \reft{e:psiSH_o}{psiSHon17} + \reft{e:psiSH_o}{psiSHon13} + \reft{e:psiSH_c}{psiSHc16} + \reft{e:psiSH_c}{psiSHc22} + \reft{e:psiSH_o}{psiSHo11}+ \reft{e:psiSH_o}{psiSHon8}+\reft{e:psiSH_e}{psiSHe18}\\
        & = [c + \iota_{\xi}v, \epsilon_n]^{(n)}X^{(a)}Z_a\underline{\mu}^\dag + [c+ \iota_{\xi}v, \epsilon_n]^{(a)} \left(L_{\xi}^{\whom}\epsilon_n + d_{\omega_n} \epsilon^n \xi^n \right)^{(n)}Z_a\underline{\mu}^\dag \\
        & \phantom{=}+ L_{\xi}^{\whom} \left([c+ \iota_{\xi}v, \epsilon_n]^{(a)}\right)Z_a\underline{\mu}^\dag
         + d_{\omega_n} \left([c+ \iota_{\xi}v, \epsilon_n]^{(a)}\right)\xi^n Z_a\underline{\mu}^\dag\\
        & \phantom{=}
         +\left( L_{\xi}^{\whom}e_a + d_{\omega_n} e_a \xi^n \right) [c+ \iota_{\xi}v, \epsilon_n]^{(a)} (Z_b\underline{\mu}^\dag)^{(b)}.
    \end{align*}
    Lastly we also have 
    \begin{align*}
        \reft{e:psiSH_e}{psiSHe19}+\reft{e:psiSH_e}{psiSHen4} = \lambda[c+ \iota_{\xi}v, \epsilon_n]^{(a)} (\nu_a +v_a) Z.
    \end{align*}

    \item \label{i:terms_with_xi}
    Similarly, for the terms containing $\xi$   we get
    \begin{align*}
    \reft{e:psiSH_e}{psiSHe1}& + \reft{e:psiSH_e}{psiSHe2} + \reft{e:psiSH_e}{psiSHe5} +\reft{e:psiSH_x}{psiSHxn2}\\
 &=\iota_{\xi}F_{\whom}  \underline{v}^\dag + L_{\xi}^{\whom}(X^{(a)}Z_a)\underline{\mu}^\dag + \frac{1}{N-3} \left(L_{\xi}^{\whom}\epsilon_n\right)^{(a)}\underline{\mu}_a^\dag e^{N-3} \left(d_{\whom}e + \tau^\dag d \xi^n\right)\\
    &\phantom{=} + \left(L_{\xi}^{\whom}\epsilon_n\right)^{(n)}X^{(a)}Z_a\underline{\mu}^\dag+  \frac{1}{N-3}L_{\xi}^{\whom}( v+\nu) e^{N-3}  \underline{\tau}^\dag
    \end{align*}
    Furthermore
    \begin{align*}
        \reft{e:psiSH_o}{psiSHon1} &+\reft{e:psiSH_o}{psiSHon15} +  \reft{e:standard_action_split3}{term4}+\reft{e:psiSH_o}{psiSHo6} + \reft{e:psiSH_c}{psiSHc1} + \reft{e:standard_action_split3}{term11}+ \reft{e:standard_action_split3}{term12}+\reft{e:psiSH_c}{psiSHc2} +\reft{e:psiSH_x}{psiSHxn16}\\
        &= L_{\xi}^{\whom}(\nu+v) e \underline{\tau}^\dag + \iota_{[\xi,\underline{z}]}(\nu+v)Z
    \end{align*}
    and 
    \begin{align*}
        \reft{e:psiSH_o}{psiSHo7}& + \reft{e:psiSH_c}{psiSHc3} + \reft{e:psiSH_x}{psiSHxn17}+ \reft{e:psiSH_x}{psiSHxn5}+ \reft{e:psiSH_x}{psiSHxn14} + \reft{e:psiSH_o}{psiSHon6} + \reft{e:psiSH_o}{psiSHon7}+ \reft{e:psiSH_c}{psiSHc8} + \reft{e:psiSH_e}{psiSHe3}\\
        &= [c + \iota_{\xi}v, \epsilon_n]^{(a)} \left(L_{\xi}^{\whom}\epsilon_n\right)^{(n)}Z_a\underline{\mu}^\dag 
        + L_{\xi}^{\whom} \left([c + \iota_{\xi}v , \epsilon_n]^{(a)}\right)Z_a\underline{\mu}^\dag\\
         & \phantom{=}
         +L_{\xi}^{\whom}e_a  [c + \iota_{\xi}v, \epsilon_n]^{(a)} (Z_b\underline{\mu}^\dag)^{(b)}+ 
         \left(L_{\xi}^{\whom}\epsilon_n\right)^{(a)} \left(d_{\omega_n} \epsilon^n \xi^n \right)^{(n)}Z_a\underline{\mu}^\dag\\
         & \phantom{=}
         + d_{\omega_n} \left(\left(L_{\xi}^{\whom}\epsilon_n\right)^{(a)}\right)\xi^n Z_a\underline{\mu}^\dag
         +d_{\omega_n} e_a \xi^n  \left(L_{\xi}^{\whom}\epsilon_n\right)^{(a)} (Z_b\underline{\mu}^\dag)^{(b)}.
    \end{align*}
    and the terms in the first line cancel with equal ones in \ref{i:terms_with_c}. Lastly
    \begin{align*}
        \reft{e:psiSH_e}{psiSHe4} + \reft{e:psiSH_e}{psiSHen1} = \iota_{[\xi,\underline{z}]}(\nu+v)Z + \lambda\left(L_{\xi}^{\whom}\epsilon_n\right)^{(a)}(\nu_a + v_a) Z.
    \end{align*}

    \item \label{i:terms_with_xi_n}
    For some of the terms containing $\xi^n$ we get 
    \begin{align*}
    \reft{e:psiSH_e}{psiSHe6}& + \reft{e:psiSH_e}{psiSHe7} + \reft{e:psiSH_e}{psiSHe10} + \reft{e:psiSH_cl}{psiSHcl3}  + \reft{e:psiSH_x}{psiSHxn8}\\
 &= F_{\omega_n}\xi^n  \underline{v}^\dag 
 + d_{\omega_n}(X^{(a)}Z_a)\underline{\mu}^\dag\xi^n + \frac{1}{N-3} \left(d_{\omega_n}\epsilon_n\xi^n\right)^{(a)}\underline{\mu}_a^\dag e^{N-3} \left(d_{\whom}e + \tau^\dag d \xi^n\right)\\
    &\phantom{=} + \left(d_{\omega_n}\epsilon_n\xi^n\right)^{(n)}X^{(a)}Z_a\underline{\mu}^\dag+  \frac{1}{N-3}d_{\omega_n}( v+\nu)\xi^n e^{N-3}  \underline{\tau}^\dag.
    \end{align*}
    The last term is compensated by
    \begin{align*}
        \reft{e:psiSH_o}{psiSHon4} &+ \reft{e:standard_action_split3}{term13} + \reft{e:standard_action_split3}{term14} + \reft{e:psiSH_c}{psiSHc4} + \reft{e:psiSH_o}{psiSHo4} \\
        &= \frac{1}{N-3}d_{\omega_n}( v+\nu)\xi^n e^{N-3}  \underline{\tau}^\dag + \iota_{\underline{z}} (\nu+v)\partial_n \xi^n \omega^\dag + \iota_{\partial_n\underline{z}}(\nu+v) \xi^n \omega^\dag.
    \end{align*}
    We also have 
    \begin{align*}
        \reft{e:psiSH_o}{psiSHo3} &+ \reft{e:psiSH_x}{psiSHxn7} + \reft{e:psiSH_x}{psiSHxn11} + \reft{e:psiSH_e}{psiSHe8} \\
        &= \left(d_{\omega_n}\epsilon_n\xi^n\right)^{(n)}X^{(a)}Z_a\underline{\mu}^\dag + d_{\omega_n}(X^{(a)}Z_a)\underline{\mu}^\dag\xi^n + [c + \iota_{\xi}v, \epsilon_n]^{(a)} \left(d_{\omega_n} \epsilon^n \xi^n \right)^{(n)}Z_a\underline{\mu}^\dag \\
        & \phantom{=}
         + d_{\omega_n} \left([c + \iota_{\xi}v, \epsilon_n]^{(a)}\right)\xi^n Z_a\underline{\mu}^\dag
         +\left(d_{\omega_n} e_a \xi^n \right) [c + \iota_{\xi}v, \epsilon_n]^{(a)} (Z_b\underline{\mu}^\dag)^{(b)}\\
        & \phantom{=}+ \left(L_{\xi}^{\whom}\epsilon_n\right)^{(a)} \left(d_{\omega_n} \epsilon^n \xi^n \right)^{(n)}Z_a\underline{\mu}^\dag
         + d_{\omega_n} \left(\left(L_{\xi}^{\whom}\epsilon_n\right)^{(a)}\right)\xi^n Z_a\underline{\mu}^\dag \\
        & \phantom{=}
         +d_{\omega_n} e_a \xi^n  \left(L_{\xi}^{\whom}\epsilon_n\right)^{(a)} (Z_b\underline{\mu}^\dag)^{(b)}
    \end{align*}
    where the second and third line in the result are canceled by the same terms appearing in items \ref{i:terms_with_c} and \ref{i:terms_with_xi} respectively. Furthermore
    \begin{align*}
        \reft{e:psiSH_e}{psiSHe9} + \reft{e:psiSH_e}{psiSHen3}  +\reft{e:psiSH_x}{psiSHxn9}&=\iota_{\underline{z}} (\nu+v)\partial_n \xi^n \omega^\dag + \iota_{\partial_n\underline{z}}(\nu+v) \xi^n \omega^\dag  \\
         &\phantom{=}+\lambda \left(d_{\omega_n}\epsilon_n \xi^n\right)^{(a)}(\nu_a + v_a) Z.
    \end{align*}

    \item Summing up items \ref{i:terms_with_c}, \ref{i:terms_with_xi} and \ref{i:terms_with_xi_n} we get
    \begin{align*}
        d_{\whom} (c + \iota_{\xi} v) + \iota_{\xi}F_{\whom}  \underline{v}^\dag + F_{\omega_n}\xi^n  \underline{v}^\dag  + \frac{1}{N-3} X^{(a)}\underline{\mu}_a^\dag e^{N-3} \left(d_{\whom}e + \tau^\dag d \xi^n\right)+ \lambda X^{(a)}(\nu_a + v_a) Z.
    \end{align*}
    The first three terms are given by \reft{e:standard_action_split3}{term21}, \reft{e:standard_action_split3}{term16}, \reft{e:psiSH_v}{psiSHv1}, \reft{e:standard_action_split3}{term15}, \reft{e:standard_action_split3}{term18} and \reft{e:standard_action_split3}{term19}. We also have that
    \begin{align*}
        \reft{e:psiSH_o}{psiSHon3}+ \reft{e:psiSH_e}{psiSHe13} + \reft{e:psiSH_x}{psiSHxn1} = X^{(a)} \mu_a^\dag d\xi^n e^{N-3} \underline{\tau}^\dag +\iota_z \left(X^{(a)} Z_a\right) d\xi^n \underline{\mu}^\dag 
    \end{align*}
    and 
    \begin{align*}
        \reft{e:psiSH_cl}{psiSHcl1}= \frac{1}{N-3} X^{(a)}\underline{\mu}_a^\dag e^{N-3} d_{\whom}e
    \end{align*}
    Furthermore
    \begin{align*}
         \reft{e:psiSH_e}{psiSHe11}+ \reft{e:psiSH_e}{psiSHe15} + \reft{e:psiSH_cl}{psiSHcl2} + \reft{e:psiSH_cl}{psiSHcl6} & = \frac{1}{N-3}e^{N-4}\lambda\nu \epsilon_n d\xi^n \underline{\tau}^\dag +\lambda \nu X^{(a)} Z_a \\
         &\phantom{=}+ \iota_z \left(X^{(a)} Z_a\right) d\xi^n \underline{\mu}^\dag + \frac{1}{N-2}e^{N-2}d_{\whom}\iota_{\underline{z}}v,
    \end{align*}
    \begin{align*}
        \reft{e:psiSH_e}{psiSHe12} + \reft{e:psiSH_e}{psiSHe14} +\reft{e:psiSH_x}{psiSHxn3} + \reft{e:psiSH_o}{psiSHon2} = \frac{1}{N-3} \lambda e^{N-4} \epsilon_n d \xi^n (\nu+ v)\underline{\tau}^\dag
    \end{align*}
    and 
    \begin{align*}
        \reft{e:standard_action_split3}{term1} = \frac{1}{N-3}e^{N-4}\lambda \epsilon_n d \xi^n v\underline{\tau}^\dag + \lambda X^{(a)}v_a Z + \frac{1}{N-2}e^{N-2}d_{\whom}\iota_{\underline{z}}v.
    \end{align*}
\end{enumerate}
Summing up all the terms we get the desired claim.

\subsubsection{Conflict of interest and data statement}
The authors have no conflicts of interest to disclose.

\noindent The manuscript has no associated data.

 \begin{refcontext}[sorting=nyt]
     \printbibliography[] 
 \end{refcontext}
\end{document}